\documentclass[11pt,times]{article}

\usepackage{amsfonts}
\usepackage{amssymb,amsmath,amsthm}
\usepackage{epsfig}
\usepackage{color}
\usepackage{latexsym}
\usepackage{enumerate}
\usepackage{fullpage,setspace}
\usepackage[dvips, paper=letterpaper, top=1in, bottom=.75in, left=1in, right=1in, nohead, includefoot, footskip=.25in]{geometry}

\title{On Oblivious PTAS's for Nash Equilibrium\footnote{Extended version of the paper of the same title that appeared in {\em STOC}, 2009.}}

\author{Constantinos Daskalakis\thanks{This work was done while the author was a postdoctoral researcher at Microsoft Research, New England.}\\EECS and CSAIL, MIT\\costis@mit.edu \and Christos Papadimitriou\thanks{Supported by NSF grant CCF - 0635319, a gift from Yahoo! Research, and a MICRO grant.}\\Computer Science, U.C. Berkeley\\christos@cs.berkeley.edu}

\usepackage{hvdashln}
\usepackage{amsfonts}
\usepackage{amssymb,amsmath}
\usepackage{epsfig}
\usepackage{color}
\usepackage{latexsym}
\usepackage{enumerate}

\newtheorem{theorem}{Theorem}
\newtheorem{lemma}{Lemma}

\newtheorem{claim} {Claim}
\newtheorem{corollary}{Corollary}
\newtheorem{prop}{Proposition}

\newtheorem{remark}{Remark}
\newtheorem{definition}{Definition}

\newenvironment{prevproof}[2]{\noindent {\bf {Proof of {#1}~\ref{#2}:}}}{$\blacksquare$\vskip \belowdisplayskip}

\newcommand{\arxiv}[1]{#1}
\newcommand{\camera}[1]{}

\def\Z{\mathcal{Z}}

\def\V{\mathcal{V}}
\def\U{\mathcal{U}}
\def\E{\mathcal{E}}
\def\Q{\mathcal{Q}}
\def\M{\mathcal{M}}

\def\G{\mathcal{G}}

\def\I{\mathcal{I}}

\def\P{{\mathcal{P}}}
\def\B{{\mathcal{B}}}
\def\A{{\mathcal{A}}}

\def\X{\mathcal{X}}
\def\Y{\mathcal{Y}}

\begin{document}
\maketitle

\begin{abstract}
If a class of games is known to have a Nash equilibrium with
probability values that are either zero or  $\Omega(1)$
--- and thus with support of bounded size --- then obviously this
equilibrium can be found exhaustively in polynomial time.
Somewhat surprisingly, we show that there is a PTAS for the class of games
whose equilibria are guaranteed to have {\em small} --- $O({1\over
n})$ --- values, and therefore large --- $\Omega(n)$ --- supports.
We also point out that there is a PTAS for games with sparse payoff
matrices, a family for which the exact problem is known to be
PPAD-complete \cite{CDT:sparse}. Both algorithms are of a special
kind that we call {\em oblivious:} The algorithm just samples a
fixed distribution on pairs of mixed strategies, and the game is
only used to determine whether the sampled strategies comprise an
$\epsilon$-Nash equilibrium; the answer is ``yes'' with inverse
polynomial probability (in the second case, the algorithm is
actually deterministic).  These results bring about the question:
{\em Is there an oblivious PTAS for finding a Nash equilibrium in
general games?} We answer this question in the negative; our lower
bound comes close to the quasi-polynomial upper bound of \cite{LMM}.

Another recent PTAS for {\em anonymous} games \cite{DP: anonymous
1,DP: anonymous 2,D: anonymous 3} is also oblivious in a weaker sense
appropriate for this class of games (it samples from a fixed
distribution on {\em unordered} collections of mixed strategies),
but its running time is exponential in $1\over \epsilon$. We prove
that any oblivious PTAS for anonymous games with two strategies and
three player types must have $1\over {\epsilon^\alpha}$ in the
exponent of the running time for some $\alpha \ge {1\over 3}$, rendering the algorithm
in~\cite{D: anonymous 3} (which works with any bounded number of
player types) essentially optimal within oblivious algorithms.  In
contrast, we devise a ${\rm poly}(n) \cdot
\left(1/\epsilon\right)^{O(\log^2({1/\epsilon}))}$ non-oblivious
PTAS for anonymous games with two strategies and any bounded number
of player types.  The key idea of our algorithm is to search not
over unordered sets of mixed strategies, but over a carefully
crafted set of collections of the first $O(\log {1\over \epsilon})$
moments of the distribution of the number of players playing
strategy $1$ at equilibrium.  The algorithm works because of a
probabilistic result of more general interest that we prove: the total
variation distance between two sums of independent indicator random variables
decreases exponentially with the number of moments of the two sums that are
equal, independent of the number of indicators.

A byproduct of our algorithm is establishing the existence of a sparse (and efficiently computable) $\epsilon$-cover of the set of all possible sums of $n$ independent indicators, under the total variation distance. The size of the cover is ${\rm poly}(n)\cdot (1/\epsilon)^{O(\log^2 (1/\epsilon))}$.

\end{abstract}

%
%


\section{Introduction}
Is there a polynomial time approximation scheme (PTAS) for computing
approximate Nash equilibria?  This has emerged, in the wake of the
intractability results for Nash equilibria \cite{DGP,CD}, as the
most central question in equilibrium computation. Over the past
three years there has been relatively slow progress towards smaller
$\epsilon$'s \cite{DMP1,DMP2,BBM,TS} --- so slow that it is hard to believe
that a PTAS is around the corner. On the other hand,
\cite{DP: anonymous 1,DP: anonymous 2} provide a PTAS for the
important special case of anonymous games with a bounded number of
strategies (those for which the utilities of the players, although
different, depend on the {\em number} of players playing each
strategy, not on the {\em identities} of the players that do). This
PTAS proceeds by discretizing the probabilities in the mixed
strategies of the players to multiples of $1\over k$, for
appropriate integer $k$, and works even in the generalization in
which the players are divided into a bounded number of {\em types},
and utilities depend on how many players {\em of each type} choose
each strategy.

In this paper we report significant progress on this important
algorithmic problem; we present several new algorithms, but also the
first nontrivial lower bounds. We start by pointing out two new
interesting classes of bimatrix games\footnote{These results can be
extended to any bounded number of players, but in what follows we
only discuss the two-player, or bimatrix, case.} that have PTAS's:

\begin{itemize}

\item It was shown in~\cite{CDT:sparse} that computing a Nash equilibrium
for the special case of {\em sparse games}, that is, games whose
payoff matrices have a bounded number of non-zero entries in each row
and column~\cite{CDT:sparse}, is PPAD-complete.  We point out that there is a trivial
PTAS --- in particular, the pair of uniform mixed strategies works.  This
is interesting in that this is the first PPAD-complete case of a
problem that is so approximable.

\item We also give a randomized PTAS for {\em small probability games,}
that is, games that are guaranteed to have Nash equilibria with
small $O({1\over n})$ nonzero probability values, and thus with
linear support (Theorem~\ref{theorem: ptas for large support games}).

\end{itemize}

It is quite surprising that games with small probability values
(our second special case above) are easy, since games with {\em
large} (bounded from below by a constant) probability values are
also easy. What probability values are difficult then? Our next
result, a lower bound, seems to suggest {\em inverse logarithmic}
probability values are hard (compare with the quasi-PTAS of
\cite{LMM}, whose equilibria have roughly logarithmic support).

To explain our negative result, we first note that both PTAS's
outlined above (as well as those for anonymous games discussed
later) are {\em oblivious}. This means that they have access to a
fixed set of pairs of mixed strategies, in the generic case by
sampling, and they look at the game only to determine whether a
sampled pair of strategies constitutes an approximate Nash
equilibrium.  Is there an oblivious PTAS for the general Nash
equilibrium? The answer here is a definite ``no'' --- in fact, we
proved our negative result after we had been working for some time
on developing such an algorithm$\ldots$  We show that any oblivious
algorithm must sample at least $\Omega\left(n^{(.8-34\epsilon)\log_2
n}\right)$ pairs in expectation (Theorem~\ref{thm:no oblivious
PTAS}). For comparison, \cite{LMM}'s algorithm takes time
$n^{O\left(\log n / \epsilon^2\right)}$.

Another important class of games for which a PTAS was recently
discovered is that of {\em anonymous games with two strategies}
\cite{DP: anonymous 1,D: anonymous 3}. These are multi-player games, in which the
utility of each player depends on the strategy (0 or 1) played, and
the {\em number} of other players playing strategy 1 ({\em not}
their identities).  In fact, the PTAS works even in the more
sophisticated case in which the players are partitioned into types,
and the utilities depend on the number of players {\em of each type}
playing strategy 1. The PTAS in \cite{DP: anonymous 1}, and the
relatively more efficient one in \cite{D: anonymous 3}, both have running
times that are exponential in $1\over \epsilon$.  They are also oblivious
in a sense appropriate for anonymous games, in that they work by
sampling an {\em unordered} set of $n$ mixed strategies, where $n$
is the number of players, and they only look at the game in order to
determine if there is an assignment of these strategies to the
players that results in an approximate equilibrium.  We prove that
any oblivious approximation algorithm, for anonymous games with two
strategies and three player types, must sample an exponential---in $1\over \epsilon$---sized collection of unordered sets of mixed strategies --- and so our PTAS's
are near-optimal.

Finally, we circumvent this negative result to develop a
non-oblivious PTAS which finds an $\epsilon$-approximate Nash
equilibrium in anonymous games with two strategies and any bounded
number of player types in time
${\rm poly}(n)\left({1\over\epsilon}\right)^{O(\log^2{1\over \epsilon})}$.
This algorithm is based (in addition to many other insights and
techniques for anonymous games) on a new result in applied
probability which is, we believe, interesting in its own right:
Suppose that you have two sums of $n$ independent Bernoulli random
variables, which have the same first moment, the same second moment, and so on up to
moment $d$. Then the distributions of the two sums have variational
distance that vanishes exponentially fast with $d$, regardless of $n$. To turn this
theorem into an algorithm, we discretize the mixed strategies of the players using techniques from~\cite{D: anonymous 3} and, in the range of parameters where the algorithm of~\cite{D: anonymous 3} breaks down, we iterate over all possible values of the first $\log (1/\epsilon)$ moments of the players' aggregate behavior; we then try to identify, via an involved dynamic programming scheme, mixed strategies, implementing the given moments, which correspond to approximate Nash equilibria of the game. Our approximation guarantee for sums of independent indicators is rather strong, especially when the number of indicators is small, a regime where Berry-Ess\'een type bounds provide weaker guarantees and result in slower algorithms~\cite{DP: anonymous 1, D: anonymous 3}. It is quite intriguing that
a quasi-polynomial time bound, such as the one we provide in this paper, shows up again in the analysis of
algorithms for approximate Nash equilibria (cf.~\cite{LMM}).

As a byproduct of our results we establish the existence of a sparse (and efficiently computable) $\epsilon$-cover of the set of all possible sums of $n$ independent indicators, under the total variation distance. The size of our cover is ${\rm poly}(n)\cdot (1/\epsilon)^{O(\log^2 (1/\epsilon))}$. We discuss this result in Section~\ref{sec:cover}.


\subsection{Preliminaries} \label{sec:prel}
A two-player, or {\em bimatrix}, game $\G$ is played by two players,
the {\em row player} and the {\em column player}. Each player has a
set of $n$ {\em pure strategies}, which without loss of generality
we assume to be the set $[n]:=\{1,\ldots,n\}$ for both players. The
game is described then by two {\em payoff matrices} $R$, $C$
corresponding to the row and column player respectively, so that, if
the row player chooses strategy $i$ and the column player strategy
$j$, the row player gets payoff $R_{ij}$ and the column player
$C_{ij}$.  As it is customary in the literature of approximate Nash
equilibria, we assume that the matrices are normalized, in that
their entries are in $[-1,1]$. 

The players can play {\em mixed strategies}, that is, probability
distributions over their pure strategies which are represented by
probability vectors $x \in \mathbb{R}_+^n$, $|x|_1=1$. If the row
player chooses mixed strategy $x$ and the column player mixed
strategy $y$, then the row player gets expected payoff
$x^{\text{\tiny T}}Ry$ and the column player expected payoff
$x^{\text{\tiny T}}Cy$.  A pair of mixed strategies $(x,y)$ is a
{\em Nash equilibrium} of the game $\G=(R,C)$ iff $x$ maximizes
$x'^{\text{\tiny T}}Ry$ among all
probability vectors $x'$ and, simultaneously, $y$ maximizes $x^{\text{\tiny T}}Cy'$ among all $y'$.  It is an {\em
$\epsilon$-approximate Nash equilibrium} iff $x^{\text{\tiny T}}Ry
\geq x'^{\text{\tiny T}}Ry-\epsilon$, for all $x'$, and, simultaneously, $x^{\text{\tiny T}}Cy
\geq x^{\text{\tiny T}}Cy'-\epsilon$, for all $y'$. In this paper, we will use the stronger notion of {\em $\epsilon$-approximately well-supported Nash equilibrium}, or simply {\em $\epsilon$-Nash equilibrium}. This is any pair of strategies $(x,y)$ such that, for all $i$ with $x_i >0$, $e_i^{\text{\tiny T}}Ry \ge e_{i'}^{\text{\tiny T}}Ry -\epsilon$, for all $i'$, and similarly for $y$. That is, every strategy $i$ in the support of $x$ guarantees expected payoff $e_i^{\text{\tiny T}}Ry$ which is within $\epsilon$ from the optimal response to $y$, and similarly every strategy in the support of $y$ is within $\epsilon$ from the optimal response to $x$.

An {\em anonymous game} is in a sense the dimensional dual of a
bimatrix game: There are $n$ players, each of which has two
strategies, 0 and 1. For each player $i$ there is a {\em utility
function } $u_i$ mapping $\{0,1\}\times [n-1]$ to $[-1,1]$.
Intuitively, $u_i(s,k)$ is the utility of player $i$ when s/he plays
strategy $s\in\{0,1\}$, and $k\leq n-1 $ of the remaining players
play strategy $1$, while $n-k-1$ play strategy $0$.  In other words,
the utility of each player depends, in a player-specific way, on the
strategy played by the player and the number of other players who
play strategy $1$ --- but {\em not} their identities. The notions of Nash equilibrium and $\epsilon$-Nash equilibrium are extended in the natural way to this setting. Briefly, a mixed strategy for the $i$-th player is a function $x_i:\{0,1\} \rightarrow [0,1]$ such that $x_i(0)+x_i(1)=1$. A set of mixed strategies $x_1,\ldots,x_n$ is then an $\epsilon$-Nash equilibrium if, for every player $i$ and every $s \in \{0,1\}$ with $x_i(s)>0$: $\E_{x_1,\ldots,x_n}u_i(s,k) \geq \E_{x_1,\ldots,x_n}u_i(1-s,k)-\epsilon$, where for the purposes of the expectation $k$ is drawn from $\{0,\ldots,n-1\}$ by flipping $n-1$ coins according to the distributions $x_j, j\neq i$, and counting the number of $1$'s.

A more sophisticated kind of anonymous games divides the players into $t$ types,
so that the utility of player $i$ depends on the strategy played by him/her, and the number of players {\em of each type} who play strategy $1$. 

\section{PTAS for Two Special Cases}

\subsection{Small Games}\label{sec:sparse}

We say that a class of bimatrix games is {\em  small} if the sum of
all entries of the $R$ and $C$ matrices is $o(n^2)$.  One such class
of games are the {\em $k$-sparse} games~\cite{CDT:sparse}, in which every row and column
of both $R$ and $C$ have at most $k$ non-zero entries.
The following result by Chen, Deng and Teng shows that
finding an exact Nash equilibrium remains hard for $k$-sparse games:

\begin{theorem}[\cite{CDT:sparse}]
Finding a $O(n^{-6})$-Nash equilibrium in $10$-sparse normalized
bimatrix games with $n$ strategies per player is a PPAD-complete
problem.
\end{theorem}

In contrast, it is easy to see that there is a PTAS for this class:

\begin{theorem} \label{thm:PTAS for sparse games}
For any $k$, there is a PTAS for the Nash equilibrium problem in
$k$-sparse bimatrix games.
\end{theorem}

Our original proof of this theorem consisted in showing that there always exists an $\epsilon$-Nash equilibrium in which both players of the game use in their mixed strategies probabilities that are integer multiples of $\epsilon/2k$. Hence, we can efficiently enumerate over all possible pairs of mixed strategies of this form, as long as $k$ is fixed. Shang-Hua Teng pointed out to us a much simpler algorithm: The pair of uniform mixed strategies is always an $\epsilon$-Nash equilibrium in a sparse game! The difference with our algorithm is this: the uniform equilibrium gives to both players payoff of at most $k/n$; our algorithm can be used instead to approximate the payoffs of the players in  the Nash equilibrium with the optimal social welfare (or more generally the Nash equilibrium that optimizes some other smooth function of the players' payoffs).

\subsection{Small Probability Games}\label{sec:large}

If all games in a class are guaranteed to have a Nash equilibrium
$(x,y)$, where the nonzero probability values in both $x$ and $y$ are larger than some constant $\delta>0$, then it is trivial to find this equilibrium in time $n^{O({1 \over\delta})}$ by exhaustive search over all possible supports and linear programming. But, what if a class of games is known to have {\em small}, say $O({1\over
n})$, probability values? Clearly, exhaustive search over supports is not efficient anymore, since those have now linear size. Surprisingly, we show that any such class of games has a
(randomized) PTAS, by exploiting the technique of~\cite{LMM}.

\begin{definition} [Small Probability Games]
A bimatrix game $\mathcal{G} = (R, C)$ is of {\em $\delta$-small
probabilities}, for some constant $\delta \in (0,1]$, if it has a Nash equilibrium
$(x,y)$ such that all the entries of $x$ and $y$ are at most $1 \over {\delta n}$.
\end{definition}

\begin{remark}
Observe that a game of $\delta$-small probabilities has an equilibrium $(x,y)$, in which both $x$ and $y$ have support of size at least $\delta n$. Moreover, there exists a subset of size at least $\frac{\delta n}{2}$ from the support of $x$, all the strategies of which have probability at least $\frac{1}{2n}$, and similarly for $y$; that is the probability mass of the distributions $x$ and $y$ spreads non-trivially over a subset of size $\Omega(n)$ of the strategies. Hence, small probability games comprise a subclass of {\em linear support games}, games with an equilibrium $(x,y)$ in which both $x$ and $y$ assign to a constant fraction $\alpha n$ of the strategies probability at least $1/\beta n$, for some constants $\alpha$ and $\beta$. However, this broader class of games is essentially as hard as the general: take any Nash equilibrium $(x,y)$ of a normalized game and define the pair $(x',y')$, where $x':= (1-{\epsilon \over 5})\cdot x+{\epsilon \over 5} \cdot \left({1 \over n},{1 \over n},\ldots,{1 \over n}\right)$ and similarly for $y'$. It is not hard to see that the new pair is an $\epsilon$-Nash equilibrium; still, regardless of what $(x,y)$ is, both $x'$ and $y'$ assign to $\alpha n$ strategies probability at least $1/\beta n$, for an appropriate selection of $\alpha$ and $\beta$.
\end{remark}

\begin{theorem} \label{theorem: ptas for large support games}
For any $\delta \in (0, 1]$, there is a randomized PTAS for normalized bimatrix games of $\delta$-small probabilities.
\end{theorem}

\begin{proof}
We show first the following (stronger in terms of the type of approximation) variant of the theorem of Lipton,
Markakis and Mehta \cite{LMM}.\camera{ The proof is given in the full version of this paper.}\arxiv{ The proof is provided in Appendix~\ref{appendix:linear support games}.}

\begin{lemma}\label{lemma: restated lipton markakis mehta}
Let $\G=(R,C)$ be a normalized bimatrix game and $(x,y)$ a Nash
equilibrium of $\G$. Let $\X$ be the distribution formed by
taking $t=\lceil 16 \log n /\epsilon^2 \rceil$ independent random
samples from $x$ and defining the uniform distribution on the
samples, and similarly define $\cal Y$ by taking samples from $y$.
Then with probability at least $1-\frac{4}{n}$ the following are
satisfied
\begin{enumerate}
\item the pair $(\X,\cal{Y})$ is an $\epsilon$-Nash equilibrium of $\G$; \label{assertion 1 in LMM restated}
\item $|e_i^{\text{\tiny T}} R \Y - e_i^{\text{\tiny T}} R y| \le \epsilon/2$, for all $i\in[n]$; \label{assertion 2 in LMM restated}
\item $|\X^{\text{\tiny T}} C e_j - x^{\text{\tiny T}} C e_j| \le \epsilon/2$, for all $j \in [n]$; \label{assertion 3 in LMM restated}
\end{enumerate}
\end{lemma}

Suppose now that we are given a normalized bimatrix game $\G=(R,C)$
of $\delta$-small probabilities, and let $(x,y)$ be an equilibrium of $\G$ in which $x_i \le {1 \over  \delta n}$, for all $i \in [n]$, and similarly for $y$.
Lemma~\ref{lemma: restated lipton markakis mehta} asserts that, if a
multiset\footnote{For our discussion, a multiset of size $t$ is an {\em ordered}
collection $\langle i_1,i_2,\ldots,i_t\rangle$ of $t$ elements from
some universe (repetitions are allowed).} $A$ of size $t=\lceil 16 \log n /\epsilon^2 \rceil$
is formed by taking $t$ independent samples from $x$ and, similarly,
a multiset $B$ is formed by taking samples from $y$, then $(\X,\Y)$,
where $\X$ is the uniform distribution over $A$ and $\Y$ the uniform
distribution over $B$, is an $\epsilon$-Nash equilibrium with
probability at least $1-4/n$. Of course, we do not know $(x,y)$ so
we cannot do the sampling procedure described above. Instead we
are going to take a uniformly random multiset $A'$ and a uniformly random multiset $B'$
and form the uniform distributions $ \X', \Y'$ over $A'$ and $B'$;
we will argue that there is an inverse polynomial chance that $(\X',
\Y')$ is actually an $\epsilon$-Nash equilibrium.

For this we define the set $\A$ of {\em good multisets for the row  player} as
$$\A : = \left\{ A ~\vline~
\begin{minipage}{6.5cm }
$A$ is a multiset, $A \subseteq [n]$, $|A| =t$, the uniform distribution $\X$ over $A$ satisfies Assertion~\ref{assertion 3 in LMM restated} of Lemma~\ref{lemma: restated lipton markakis mehta}
\end{minipage}
\right\},$$
and, similarly, the set $\B$ of {\em good multisets of the column player} as
$$\B : = \left\{ B ~\vline~
\begin{minipage}{6.5cm }
$B$ is a multiset, $B \subseteq [n]$, $|B| =t$, the uniform distribution $\Y$ over $B$ satisfies Assertion~\ref{assertion 2 in LMM restated} of Lemma~\ref{lemma: restated lipton markakis mehta}
\end{minipage}
\right\}.$$
The reason for defining $\A$ and $\B$ in this way is that, given two multisets $A \in \A$, $B \in \B$, the uniform distributions $\X$ over $A$ and $\Y$ over $B$ comprise an $\epsilon$-Nash equilibrium\camera{ (see the full version of the paper for a detailed justification).}\arxiv{ (see the proof of Lemma~\ref{lemma: restated lipton markakis mehta} for a justification).}

What remains to show is that, with inverse polynomial probability, a random multiset $A'$ belongs to $\A$ and a random multiset $B'$ belongs to $\B$. To show this we lower bound the cardinalities of the sets $\A$ and $\B$ via the following claim,\camera{ proven in the full version of this paper.}\arxiv{ proven in Appendix~\ref{appendix:linear support games}.} We argue that the subset of $\A$ containing elements that could arise by sampling $x$ is large: indeed, with probability at least $1-{4 \over n}$, a multiset sampled from $x$ belongs to $\A$ and, moreover, each individual multiset has small probability of being sampled, since $x$ spreads the probability mass approximately evenly on its support. 

\begin{claim} \label{claim:cardinality of good multisets is large}
The sets $\A$ and $\B$ satisfy
$$|\A| \ge \left(1-\frac{4}{n}\right)\left( \delta n \right)^t \text{~and~} |\B| \ge  \left(1-\frac{4}{n}\right)\left( \delta n \right)^t.$$
\end{claim}
Given Claim~\ref{claim:cardinality of good multisets is large}, we can show Claim~\ref{claim: probability of success of linear support large}; the proof is given\camera{ in the full version of this paper.}\arxiv{ in Appendix~\ref{appendix:linear support games}.} Equation \eqref{eq:large support games final probability} implies that the algorithm that samples two uniformly random multisets $A'$, $B'$ and forms the uniform probability distributions $\X'$ and $\Y'$ over $A'$ and $B'$ respectively, succeeds in finding an $\epsilon$-Nash equilibrium with probability inverse polynomial in $n$. This completes the proof of Theorem~\ref{theorem: ptas for large support games}.
\begin{claim}\label{claim: probability of success of linear support large} If $\X', \Y'$ are the uniform distributions over random multisets $A'$ and $B'$ then
\begin{align}
&\Pr \left[ (\X',\Y') \text{ is an $\epsilon$-Nash equilibrium} \right] \notag\\&~~~~~~~~~~~~~~~~~ =  \Omega \left(\delta^2 \cdot n^{-32 \log(1/\delta)/{\epsilon^2}}\right). \label{eq:large support games final probability}
\end{align}
\end{claim}
\end{proof}
\section{A Lower Bound for Bimatrix Games} \label{sec:lower bound}

The two PTAS's presented in the previous section are {\em oblivious}.  An
oblivious algorithm looks at the game only to check if the
various pairs of mixed strategies it has come up with (by
enumeration or, more generally, by random sampling) are
$\epsilon$-approximate, and is guaranteed to come up with one that
is with probability at least inverse polynomial in the game
description.   More formally, an oblivious algorithm for the Nash
equilibrium problem is a distribution over pairs of mixed
strategies, indexed by the game size $n$ and the desired
approximation $\epsilon$.  It is a PTAS if for any game the
probability that a pair drawn from the distribution is an
$\epsilon$-Nash equilibrium is inversely polynomial in $n$. Notice
that, since we are about to prove lower bounds, we are opting for
the generality of randomized oblivious algorithms---a deterministic
algorithm that enumerates over a fixed set of pairs of mixed strategies can be seen as a (randomized) oblivious algorithm by considering the uniform distribution over the set it enumerates over.

The rather surprising simplicity and success of these algorithms (as
well as their cousins for anonymous games, see the next section)
raises the question: Is there an oblivious PTAS for the general Nash
equilibrium problem? We show that the answer is negative.

\begin{theorem}  There is no oblivious PTAS for the Nash equilibrium in bimatrix games. \label{thm:no oblivious PTAS}
\end{theorem}
\begin{proof} We construct a super-polynomially large family of $n \times n$ games with the property that every two games in the family do not share an $\epsilon$-Nash equilibrium. This quickly leads to the proof of the theorem.

Our construction is based on a construction by Alth\"ofer~\cite{Alt}, except that we need to pay more attention to ensure that the $\epsilon$-Nash equilibria of the games we construct are ``close'' to the exact Nash equilibria. For  $\ell$ even and $n={\ell \choose \ell/2}$, we define
a family of $n\times n$ two-player games $\G_S=(R_S,C_S)$, indexed
by all subsets $S\subseteq [n]$ with $|S|=\ell$. Letting $\{ S_1, S_2, \ldots, S_n \}$ be the set of all subsets of $S$ with cardinality $\ell /2$, we imagine that column $j$ of the game $\G_S$ corresponds to subset $S_j$. Then, for every $j$, we fill column $j$ of the payoff matrices $R_S$ and $C_S$ as follows:
\begin{itemize}
\item for all $i \notin S$, $R_{S,ij} = -1$ and $C_{S,ij} = 1$;

\item for all $i \in S_j$, $R_{S,ij} = 1$ and $C_{S,ij} = 0$; and

\item for all $i \in S \setminus S_j$, $R_{S,ij} = 0$ and $C_{S,ij} = 1$.

\end{itemize}
In other words, our construction has two components: In the first component (defined by the rows labeled with the elements of $S$), the game is $1$-sum, whereas in the second (corresponding to the complement of $S$) the game is $0$-sum with the row player always getting payoff of $-1$ and the column player always getting payoff of $1$. The payoffs of the first component are more balanced in the following way: as we said, every column corresponds to a subset $S_1\ldots,S_n$ of $S$ of cardinality $\ell/2$; if the column player chooses column $j$, then the row player gets $1$ for choosing a row corresponding to an element of $S_j$ and $0$ for a row corresponding to an element in $S\setminus S_j$.\arxiv{ See Figure~\ref{fig: the payoff matrix RS} of Appendix \ref{appendix:lower bound for bimatrix} for an illustration of $R_S$ for the case $n=6$, $\ell=4$, $S=\{1,2,3,4\}$.}

Lemma~\ref{lem: claim about the properties of strategy of the row player at eps-equilibrium}\camera{ (see proof in the full version of the paper)} provides the following characterization of the approximate equilibria of the game $\G_S$: in any $\epsilon$-Nash equilibrium $(x,y)$, the strategy $x$ of the row player must have $\ell_1$ distance at most $8\epsilon$ from $u_S$ ---the uniform distribution over the set $S$. That is, in all approximate equilibria of the game the strategy of the row player must be close to the uniform distribution over $S$. Formally,

\begin{lemma} \label{lem: claim about the properties of strategy of the row player at eps-equilibrium}
Let $\epsilon <1$. If $(x,y)$ is an $\epsilon$-Nash equilibrium of the game $\G_S$, where $x$ is the mixed strategy of the row player and $y$ that of the column player, then the following properties are satisfied by $x$:
\begin{enumerate}
\item $x_i = 0$, for all $i \notin S$;

\item $\ell_1(x, u_S) \le 8 \epsilon$, where $u_S$ is the uniform distribution over $S$, and $\ell_1(x, u_S)$ represents the $\ell_1$ distance between distributions $x$, $u_S$.
\end{enumerate}
\end{lemma}

The proof of the first assertion is straightforward: the row player will not assign any probability mass to the rows which give her $-1$, since any row in $S$ will guarantee her at least $0$. Since all the activity happens then in the first component of the game, which is $1$-sum, an averaging argument implies that both players' payoff is about $1/2$ at equilibrium. Observe further that, for a given mixed strategy of the row player, the strategy of the column player that guarantees her the highest payoff is the subset containing the $\ell/2$ elements of $S$ to which the row player assigns the smallest probability mass. Hence, if the probability distribution of the row player were far from uniform, then, contrary to what we argued, the corresponding payoff for the column player would be larger than $1/2$---this is established via a delicate geometric argument for $\ell_1$ distances of probability measures.\arxiv{ See Lemma~\ref{lem: bound on l1 from bound on the sum of the half of them} in Appendix~\ref{appendix:lower bound for bimatrix}.}

Suppose now that there is an oblivious PTAS for the Nash equilibrium, that is, a distribution $D$ over pairs of mixed strategies such that, for any game $\G_S$, the probability that an $\epsilon$-approximate Nash equilibrium is sampled is inverse polynomial in $n$. Let us  consider the probability distribution $D_R$ induced by $D$ on the mixed strategies of the row player and denote by $B_S$ the $\ell_1$ ball of radius $8 \epsilon$ around $u_S$. Lemma~\ref{lem: claim about the properties of strategy of the row player at eps-equilibrium} implies that $D_R$ should be assigning probability mass at least inverse polynomial in $n$ to each ball $B_S$, $S \subseteq [n]$, $|S| = \ell$. This is impossible, since by the following claim there is a super-polynomial number of disjoint such balls. The proof of the claim is via a counting argument\camera{ given in the full version of this paper}.\arxiv{ See Appendix~\ref{appendix:lower bound for bimatrix}.}
\begin{claim}\label{claim:superpolynomially many disjoint balls}
There is a family of $\Omega\left(
n^{(.8-34\epsilon)\log_2 n}\right)$ disjoint such balls.
\end{claim}\end{proof}

The proof of Theorem~\ref{thm:no oblivious PTAS} implies in particular that any oblivious PTAS for general two-player games requires expected running time of at least $\Omega\left(
n^{(.8-34\epsilon)\log_2 n}\right)$. Compare this bound to the $n^{O(\log n/\epsilon^2)}$ upper bound obtained by Lipton, Markakis and Mehta~\cite{LMM}.

\section{A Lower Bound for Anonymous Games}
Recall the definition of anonymous games from Section~\ref{sec:prel}.  In
\cite{DP: anonymous 1} we give a PTAS for anonymous games with two strategies, running in time $n^{O(({1/
\epsilon})^2)}$,  and in
\cite{D: anonymous 3} a more efficient one, with running time $n^{O(1)} \cdot (1/\epsilon)^{O(({1/
\epsilon})^2)}$. (In
\cite{DP: anonymous 2} we also give a much more sophisticated PTAS
for anonymous games with more than two strategies).  All these
PTAS's have $1\over \epsilon$ in the exponent of the running time, and they work even if
there is a fixed number $t$ of types (in which case $t$ multiplies
the exponent). Furthermore, it turns out that all of these
algorithms are {\em oblivious}, in a specialized sense appropriate
for anonymous games defined next.

{\em An oblivious $\epsilon$-approximation algorithm for
anonymous games with $n$ players} is defined in terms of a
distribution, indexed by $n$,  on {\em unordered
$n$-tuples of mixed strategies}.\footnote{For ordered sets of mixed strategies, the lower bound we are about to show becomes trivial and uninteresting.} The algorithm samples from this
distribution, and for each $\{p_1,\ldots, p_n\}$ sampled, it
determines whether there is an assignment of these probabilities to
the $n$ players such that the resulting strategy profile (with each
player playing strategy 1 with the assigned probability) is an
$\epsilon$-approximate Nash equilibrium; this latter test can be
carried out by max-flow techniques, see, e.g.,  \cite{DP: anonymous
1}.  The expected running time of this approximation algorithm is then the
inverse of its probability of success.

We show the following result, implying that any oblivious $\epsilon$-approximation algorithm for anonymous games whose expected running time is polynomial in the number of players must have expected running time exponential in $({1 \over \epsilon})^{1/3}$. Hence, our PTAS from~\cite{D: anonymous 3} is essentially optimal among oblivious PTAS's.

\begin{theorem} \label{theorem: no subexponential oblivious PTAS}
For any constant $c\ge 0$, $\epsilon <1$, no oblivious $\epsilon$-approximation algorithm for anonymous games with $2$ strategies and $3$ player types has probability of success larger than $n^{-c} \cdot 2^{-\Omega(1 /
\epsilon^{1/3})}$.
\end{theorem}

We only sketch the proof next and postpone the details for 
Appendix~\ref{appendix:lower bound for anonymous}. We first show the following\arxiv{ (see Theorem~\ref{theorem: anonymous games with unique equilibria up to plus minus small} in Appendix~\ref{appendix:lower bound for anonymous:games with unique equilibria})}: given any ordered $n$-tuple $(p_1,\ldots,p_n)$ of probabilities, we can construct a normalized anonymous game with $n$ players of type A, and two more players of their own type, such that in any $\epsilon$-Nash equilibrium of the game the $i$-th player of type A plays
strategy 1 with probability very close (depending on $\epsilon$ and $n$) to $p_i$. To obtain this game, we need to understand how to exploit the difference in the payoff functions of the players of type A to enforce different behaviors at equilibrium, despite the fact that in all other aspects of the game the players of group A are indistinguishable.

The construction is based on the following idea: For all $i$, let us denote by $\mu_{-i}:=\sum_{j \neq i}p_j$ the {\em target} expected number of type-A players different than $i$ who play strategy $1$; and let us give this payoff to player $i$ if she plays strategy $0$, regardless of what the other players are doing. If $i$ chooses $1$, we give her payoff $t$ where $t$ is the number of players different than $i$ who play $1$. By setting the payoffs in this way we ensure that $(p_1,\ldots,p_n)$ is in fact an equilibrium, since for every player the payoff she gets from strategy $0$ matches the expected payoff she gets from strategy $1$. However,  enforcing that $(p_1,\ldots,p_n)$ is also the unique equilibrium is a more challenging task; and to do this we need to include two other players of their own type: we use these players to ensure that the sum of the mixed stategies of the players of type A matches $\sum p_i$ at equilibrium, so that a player $i$ deviating from her prescribed strategy $p_i$ is pushed back towards $p_i$.
We show how this can be done in Appendix~\ref{appendix:lower bound for anonymous:games with unique equilibria}. We also provide guarantees for the $\epsilon$-Nash equilibria of the resulting game. 

The construction outlined above enables us to define a family of $2^{\Omega(1 / \epsilon^{1/3})}$ anonymous games with the property that no two games in the family share an $\epsilon$-Nash equilibrium, even as an unordered set of mixed strategies\arxiv{ (Claims~\ref{claim:lower bound: no two games share an equilibrium} and \ref{claim:lower bound:many sets})}. Then, by an averaging argument, we can deduce that for any oblivious algorithm there is a game in the ensemble for which the probability of success is at most $2^{-\Omega(1 / \epsilon^{1/3})}$. It is important for our construction to work that the anonymous game defined for a given collection of probability values $(p_1,p_2,\ldots,p_n)$ does not deviate too much from the prescribed set of mixed strategies $p_1,p_2,\ldots,p_n$ even in an $\epsilon$-Nash equilibrium. The bound of $2^{\Omega(1 / \epsilon^{1/3})}$ emerges from a quantification of this deviation as a function of $\epsilon$. The proof of Theorem~\ref{theorem: no subexponential oblivious PTAS} is given in Appendix~\ref{appendix:lower bound for anonymous:proof of the lower bound}.

\begin{remark}
We can show an equivalent of Theorem~\ref{theorem: no subexponential oblivious PTAS} for oblivious $\epsilon$-approximation algorithms for anonymous games with $2$-player types and $3$ strategies per player. The details are omitted.
\end{remark}

\section{A Quasi-polynomial PTAS}

We circumvent the lower bound of the previous section by
providing a PTAS for anonymous games with running time polynomial in the number of players
$n$ times a factor of $({1\over \epsilon})^{O(\log^2 {1\over\epsilon})}$.  The
PTAS is, of course, non-oblivious, and in fact in the following
interesting way:  Instead of enumerating a fixed set of
unordered collections of probability values, we enumerate a fixed set of {\em  $\log(1/\epsilon)$-tuples, representing the first $\log(1/\epsilon)$ moments of these probability values.} We can think of the these moments as {\em more succinct aggregates} of mixed strategy profiles than the unordered collections of strategies considered in~\cite{D: anonymous 3,DP: anonymous 1}, since several of these collections may share the same moments. To put the idea into context, let us recall the following theorem.
\begin{theorem}[\cite{D: anonymous 3, D: anonymous 3 full}] \label{theorem: structural result for approximate equilibria}
For every $\epsilon >0$, there exists some integer $k=O(1/\epsilon)$ such
that for any $n$-player $2$-strategy anonymous game with payoffs in
$[-1,1]$ there is an $\epsilon$-Nash equilibrium such that
\begin{enumerate}
\item either at most $k^3$ players randomize, and their mixed strategies are integer multiples of $1/k^2$; \label{item in structural result: a few players randomize}
\item or all players who randomize use the same mixed strategy, and this strategy is an integer multiple of $\frac{1}{kn}$.\label{item in structural result: many players randomize}
\end{enumerate}
\end{theorem}
This structural result can be turned into an oblivious PTAS using max-flow arguments\arxiv{ (see Appendix~\ref{appendix: algorithm from wine} for details)}\camera{ (see~\cite{D: anonymous 3})}. At its heart the proof of the theorem relies on the following intuitive fact about sums of indicator random variables: {\em If two sums of independent indicators have close means and variances, then their total variation distance should be small.} The way this fact becomes relevant to anonymous games is that, if there are $2$ strategies per player, then the mixed strategy of a player can be described by an indicator random variable; and as it turns out, if we replace one set of indicators by another, the change in payoff that every player will experience is bounded by the total variation distance between the sum of the indicators before and after the change. Nevertheless, the bound obtained by approximating the first two moments of the sum of the Nash equilibrium strategies is weak, and the space of unordered sets of probability values that we need to search over becomes exponential in $1/\epsilon$.

To obtain an exponential pruning of the search space, we turn to higher moments of the Nash equilibrium. We show the following theorem, which provides a rather strong quantification of how the total variation distance between two sums of indicators depends on the number of their first moments that are equal.

\begin{theorem} \label{theorem:binomial appx theorem}
Let $\mathcal{P}:=(p_i )_{i=1}^n \in (0,1/2]^n$ and
$\mathcal{Q}:=(q_i)_{i=1}^n \in (0,1/2]^n$ be two collections of
probability values in $(0,1/2]$. Let also
$\mathcal{X}:=(X_i)_{i=1}^n$ and $\mathcal{Y}:=(Y_i)_{i=1}^n$ be two
collections of independent indicators with $\E[X_i]=p_i$ and
$\E[Y_i]=q_i$, for all $i \in [n]$. If for some $d \in [n]$ the
following condition is satisfied:
$$(C_d):~~\sum_{i=1}^n p_i^{\ell} = \sum_{i=1}^n q_i^{\ell},~~~\text{for all } \ell=1,\ldots,d,$$
\begin{align}
\text{then}~~\left|\left|\sum_{i}{X_i}~;~\sum_{i}{Y_i} \right|\right|\le
20(d+1)^{1/4} 2^{-(d+1)/2}.~~\label{eq: target equation
X,Y}\end{align}
\end{theorem}
\noindent Condition $(C_d)$ considers the power sums of the expectations of the indicator random variables. Using the theory of symmetric polynomials we can show that $(C_d)$ is equivalent to the following condition on the moments of the sums of the indicators\camera{ (for the proof see the full version of the paper)}\arxiv{ (for the proof see Proposition~\ref{proposition:variable moments to probability moments} in Appendix~\ref{appendix:binomial approximation theorem})}:
$$(V_d):~~\E\left[\left(\sum_{i=1}^n X_i\right)^{\ell}\right] = \E\left[\left(\sum_{i=1}^n Y_i\right)^{\ell}\right],~\text{for all } \ell \in [d].$$
Theorem~\ref{theorem:binomial appx theorem} provides then the following strong approximation guarantee for sums of indicator random variables, that we think should be important in other settings. Our result is related to the classical moment method in probability theory~\cite{durett}, but to our knowledge it is novel and significantly stronger than known results:

\medskip\begin{minipage}{15cm}{\em If two sums of independent indicators with expectations bounded by 1/2 have equal first $d$ moments, then their total variation distance is $2^{-\Omega(d)}$.} \end{minipage}

\medskip \noindent It is important to note that our bound in~\eqref{eq: target equation X,Y} does not rely on summing up a large number of indicators $n$. This is quite critical since the previous techniques break down for small $n$'s---for large $n$'s Berry-Ess\'een type bounds are sufficient to obtain strong guarantees (this is the heart of the probabilistic results used in~\cite{DP: anonymous 1, D: anonymous 3}). 

The proof of Theorem~\ref{theorem:binomial appx theorem} (and its complement for probability values in $[1/2,1)$) is given in\camera{ the full version of this paper.}\arxiv{ Appendix~\ref{appendix:binomial approximation theorem}.} It proceeds by expressing the distribution of the sum of $n$ indicators, with expectations $p_1,\ldots,p_n$, as a weighted sum of the binomial distribution $\B(n,p)$ (with $p = \bar{p}= {\sum p_i / n}$) and its $n$ derivatives with respect to $p$, at the value $p=\bar{p}$ (these derivatives correspond to finite signed measures). It turns out that the coefficients of the first $d$ terms of this expansion are symmetric polynomials with respect to the probability values $p_1,\ldots,p_n$, of degree at most $d$; hence, from the theory of symmetric polynomials, each of these coefficients can be written as a function of the power-sum symmetric polynomials $\sum_i p_i^{\ell}$, $\ell=1,\ldots,d$~(see, e.g.,~\cite{Zolotarev}). So, if two sums of indicators satisfy Condition~$(C_d)$, the first terms cancel, and the total variation distance of the two sums depends only on the other terms of the expansion (those corresponding to higher derivatives of the binomial distribution). The proof is concluded by showing that the joint contribution of these terms is inverse polynomial in $2^{\Omega(d)}$.\arxiv{ (The details are given in Appendix~\ref{appendix:binomial approximation theorem}).}
%

Our algorithm, shown below, exploits the strong approximation guarantee of Theorem~\ref{theorem:binomial appx theorem} to improve upon the algorithm of~\cite{D: anonymous 3} in the case where only $O(1/\epsilon^3)$ players mix at the $\epsilon$-Nash equilibrium (this corresponds to case 1 in Theorem~\ref{theorem: structural result for approximate equilibria}). The complementary case (case 2 in Theorem~\ref{theorem: structural result for approximate equilibria}) can be treated easily in time polynomial in $n$ and $1/\epsilon$ by exhaustive search and max-flow arguments\arxiv{ (see Appendix~\ref{appendix: algorithm from wine}).}\camera{ (see~\cite{D: anonymous 3}).}

\medskip\noindent {\sc Algorithm Moment Search}\text{}\\
{\bf Input:}An anonymous game $\cal G$, the desired approximation
$\epsilon$.\\
{\bf Output:} An $\epsilon$-Nash equilibrium of $\cal G$ in which
all probabilities are integer multiples of $1\over k^2$, where
$k=\lceil{c\over \epsilon}\rceil$ and $c$ is universal (independent of $n$) constant, determined by Theorem~\ref{theorem: structural result for approximate equilibria}. For technical reasons that will be clear shortly, we choose a value for $k$ that is by a factor of $2$ larger than the value required by Theorem~\ref{theorem: structural result for approximate equilibria}; this is the value $k$ that guarantees an $\epsilon$/2-Nash equilibrium in multiples of $1/k^2$. Finally, we assume that we have already performed the search corresponding to Case 2 of Theorem~\ref{theorem: structural result for approximate equilibria} for this value of $k$ and we have not found an $\epsilon/2$-Nash equilibrium. So there must exist an $\epsilon/2$-Nash equilibrium in which at most $k^3$ players randomize in integer multiples of $1/k^2$.
\begin{enumerate}
\item Guess integers $t_0,t_1,t_s,t_b\leq n$, $t_s+t_b \le k^3$, where $t_0$ players
will play pure strategy $0$, $t_1$ will play pure strategy $1$,
$t_s$ will mix with probability $\leq {1\over 2}$, and $t_b = n -
t_0 - t_1 - t_s$ will mix with probability $>{1 \over 2}$.  (Note that we have to handle low and high
probabilities separately, because Theorem~\ref{theorem:binomial appx
theorem} only applies to indicators with expectations in $(0,1/2]$; we handle indicators with expectations in $(1/2,1)$ by taking their complement and employing Theorem \ref{theorem:binomial appx
theorem}\arxiv{---see Corollary~\ref{theorem:binomial appx theorem for large guys}}\camera{---see the full version for the precise details}.) \label{item in Moment Search: guessing how many small how many large}

\item For $d = \lceil 3 \log_2(320/\epsilon) \rceil$, guess $\mu_{1}$, $\mu_{2}$, $\ldots$, $\mu_d$, $\mu'_{1}$, $\mu'_{2}$, $\ldots$, $\mu'_d$, where, for all $\ell \in [d]$:
$$\mu_{\ell} \in \left\{ j \left(\frac{1}{k^2}\right)^{\ell}:~t_s \le j \leq t_s
\left(\frac{k^2}{2} \right)^{\ell} \right\},$$ and $$\mu'_{\ell} \in \left\{ j
\left(\frac{1}{k^2}\right)^{\ell}: t_b \left(\frac{k^2}{2}+1 \right)^{\ell}\le j \leq t_b (k^2 - 1)^{\ell} \right\}.$$ 

\smallskip For all $\ell$, $\mu_{\ell}$ represents the $\ell$-power sum of the mixed strategies  of the players who mix and choose strategies from the set $\{1/k^2,\ldots,1/2\}$. Similarly, $\mu'_{\ell}$ represents the $\ell$-power sum of the mixed strategies  of the players who mix and choose strategies from the set $\{1/2+1/k^2,\ldots,(k^2-1)/k^2\}$. {\bf Remark:} Whether there actually exist probability values $\pi_1,\ldots, \pi_{t_s} \in \{1/k^2,\ldots,1/2\}$ and $\theta_1,\ldots,\theta_{t_b} \in \{1/2+1/k^2,\ldots,(k^2-1)/k^2\}$ such that $\mu_{\ell}=\sum_{i=1}^{t_s} \pi_i^{\ell}$ and $\mu'_{\ell}=\sum_{i=1}^{t_b} \theta_i^{\ell}$, for all $\ell=1,2,\ldots,d$, will be determined later.  \label{item
in Moment Search: guessing moments}

\item For each player $i=1,\ldots, n$, find a subset $$\mathcal{S}_i
\subseteq \left\{0,
\frac{1}{k^2},\ldots,\frac{k^2-1}{k^2},1\right\}$$ of permitted
strategies for that player in an $\epsilon\over 2$-Nash
equilibrium, conditioned on the guesses in the previous steps. By this, we mean determining the answer to the following: ``Given our guesses for the aggregates $t_0$, $t_1$, $t_s$, $t_b$, $\mu_{\ell}$, $\mu'_{\ell}$, for all $\ell \in [d]$, what multiples of $1/k^2$ could player $i$ be playing in an $\epsilon/2$-Nash equilibrium?'' Our test exploits the anonymity of the game and uses Theorem \ref{theorem:binomial appx theorem} to achieve the following: 
\begin{itemize}

\item if a multiple of $1/k^2$ can be assigned to player $i$ and complemented by choices of multiples for the other players, so that the aggregate conditions are satisfied and player $i$ is at $3\epsilon/4$-best response (that is, she experiences at most $3\epsilon/4$ regret), then this multiple of $1/k^2$ is included in the set $\mathcal{S}_i$;

\item if, given a multiple of $1/k^2$ to player $i$, there exists no assignment of multiples to the other players so that the aggregate conditions are satisfied and player $i$ is at $3\epsilon/4$-best response, the multiple is rejected from set $\mathcal{S}_i$.
\end{itemize}
Observe that the value of $3\epsilon/4$ for the regret used in the classifier is intentionally chosen midway between $\epsilon/2$ and $\epsilon$. The reason for this value is that, if we only match the first $d$ moments of a mixed strategy profile, our estimation of the real regret in that strategy profile is distorted by an additive error of $\epsilon /4$ (coming from~\eqref{eq: target equation X,Y} and the choice of $d$). Hence, with a threshold at $3\epsilon/4$ we make sure that: a. we are not going to ``miss'' the $\epsilon/2$-Nash equilibrium (that we know exists in multiples of $1/k^2$ by virtue of our choice of a larger $k$), and b. {\em any} strategy profile that is consistent with the aggregate conditions and the sets $\mathcal{S}_i$ found in this step is going to have regret at most $3\epsilon/4+ \epsilon/4 =\epsilon$.\camera{ The fairly involved details of our test, and the way its analysis ties in with the search for an $\epsilon$-Nash equilibrium is given in the full version of this paper.}\arxiv{ The fairly involved details of our test are given in Appendix~\ref{appendix: classifier for moment search}, and the way its analysis ties in with the search for an $\epsilon$-Nash equilibrium is given in the proofs of Claims~\ref{claim: moment search never fails} and~\ref{claim: moment search always outputs equilibrium} of Appendix~\ref{appendix: analysis of moment search}.}

 \label{item in Moment Search: determining the sets Si}

\item Find an assignment of mixed strategies $v_1 \in \mathcal{S}_1$,
$\ldots$, $v_n \in \mathcal{S}_n$ to players, such that:\label{item
in Moment Search: check if there is an assignment from the Sis
implementing the moments}
\begin{itemize}
\item $t_0$ players are assigned value $0$ and $t_1$ players are assigned value $1$;
\item $t_s$ players are assigned a value in $(0,1/2]$ and $\sum_{i:v_i \in (0,1/2]} v_i^{\ell} = \mu_{\ell}$, for all $\ell \in [d]$;
\item $t_b$ players are assigned a value in $(1/2,1)$ and $\sum_{i:v_i \in (1/2,1)} v_i^{\ell} = \mu'_{\ell}$, for all $\ell \in [d]$.
\end{itemize}
Solving this assignment problem is non-trivial, but it can be done by dynamic programming in time $$O(n^3) \cdot \left(1\over \epsilon\right)^{O(\log^2(1/\epsilon))},$$ because the sets $\mathcal{S}_i$ are subsets of $\{0,1/k^2,\ldots,1\}$. The algorithm is given in the\camera{ full version of the paper.}\arxiv{ proof of Claim~\ref{claim: easy to solve moment equations} in Appendix~\ref{appendix: analysis of moment search}.}

\item  If an assignment is found, then the vector $(v_1,\ldots,v_n)$ constitutes an $\epsilon$-Nash equilibrium. \label{item in Moment Search: output a mixed strategy profile}
\end{enumerate}

\begin{theorem}
{\sc Moment Search} is a PTAS for $n$-player $2$-strategy
anonymous games with running time $U\cdot {\rm poly}(n) \cdot
(1/\epsilon)^{O( \log^2 (1/\epsilon))},$ where $U$ is the number of
bits required to represent a payoff value of the game. The algorithm
generalizes to a constant number of player types with the number of
types multiplying the exponent of the running time.
\end{theorem}

\noindent{\bf Sketch:}  
Correctness follows from this observation:  The results in
\cite{D: anonymous 3} and the choice of $k$ guarantee that an $\epsilon\over 2$-approximate
Nash equilibrium in discretized probability values exists;
therefore, Step 3 will find non-empty ${\cal S}_i$'s for all players
(for some guesses in Steps 1 and 2, since in particular the $\epsilon/2$-Nash equilibrium will survive the tests of Step 3---by Theorem
\ref{theorem:binomial appx theorem} and the choice of $d$, at most $\epsilon/4$ accuracy is lost if the correct values for the moments are guessed); and thus Step 4 will find an $\epsilon$-approximate Nash equilibrium (another $\epsilon/4$ might be lost in this step). The full proof and the running time analysis are provided\arxiv{ in Appendix~\ref{appendix: analysis of moment search}.}\camera{ in the full version of this paper.}

\section{A Sparse $\epsilon$-cover for Sums of Indicators} \label{sec:cover}

A byproduct of our proof is showing the existence of a sparse (and efficiently computable) $\epsilon$-cover of the set of sums of independent indicators, under the total variation distance. To state our cover theorem, let ${\cal S}:= \{ \{X_i\}_i~|~X_1,\ldots,X_n~\text{are independent indicators} \}$. 
We show the following.

\begin{theorem} [Cover for sums of indicators] \label{thm: sparse cover theorem}
For all $\epsilon >0$, there exists a set ${\cal S}_{\epsilon} \subseteq {\cal S}$ such that
(i)
$|{\cal S}_{\epsilon}| \le n^3 \cdot O(1/\epsilon) + n \cdot \left({1 \over \epsilon}\right)^{O(\log^2{1/\epsilon})}$;
(ii) For every $\{X_i\}_i \in {\cal S}$ there exists some $\{Y_i\}_i  \in {\cal S}_{\epsilon}$ such that $d_{\rm TV}(\sum_i
X_i,\sum_i Y_i) \le \epsilon$; and
(iii) the set ${\cal S}_{\epsilon}$ can be constructed in time $O\left(n^3 \cdot O(1/\epsilon) + n\cdot \left({1 \over
\epsilon}\right)^{O(\log^2{1/\epsilon})}\right)$.
Moreover, if  $\{Y_i\}_i  \in {\cal S}_{\epsilon}$, then the collection $\{Y_i\}_i$ has one of the following forms, where
$k=k(\epsilon) = O(1/\epsilon)$ is a positive integer: 
\begin{itemize}
\item (Sparse Form) There is a value $ \ell \leq k^3=O(1/\epsilon^3)$
such that
for all $i \leq \ell$ we have $\E[{Y_i}] \in \left\{{1 \over k^2}, {2\over k^2},\ldots, {k^2-1 \over k^2 }\right\}$, and
for all $i >\ell $ we have $\E[{Y_i}] \in \{0,  1\}$.
\item ($k$-heavy Binomial Form) There is a value $\ell \in \{0,1,\dots,n\}$
and a value $q \in \left\{ {1 \over kn}, {2 \over kn},\ldots, {kn-1 \over kn}
\right\}$ such that
for all $i \leq \ell$ we have $\E[{Y_i}] = q$;
for all $i >\ell$ we have $\E[{Y_i}] \in \{0,  1\}$; and $\ell,q$ satisfy the bounds
$\ell q \ge k^2-{1\over k}$ and
$\ell q(1-q) \ge k^2- k-1-{3\over k}.$
\end{itemize}
\end{theorem}

\begin{prevproof}{Theorem}{thm: sparse cover theorem}
Daskalakis~\cite{D: anonymous 3 full} establishes the same theorem, except that the size of the cover he produces, as well as the time needed to produce it, are  $n^3 \cdot O(1/\epsilon) + n \cdot \left({1 \over \epsilon}\right)^{O({1/\epsilon^2})}$. Indeed, this bound is obtained by enumerating over all possible collections $\{Y_i\}_i$ in sparse form and all possible collections in $k$-heavy Binomial Form, for $k=O(1/\epsilon)$ specified by the theorem. Indeed, the total number of collections in $k$-heavy Binomial form is at most $(n+1)^2 n k=n^3 \cdot O(1/\epsilon)$, since there are at most $n+1$ choices for the value of $\ell$, at most $k n$ choices for the value of $q$, and at most $n+1$ choices for the number of variables indexed by $i > \ell$ that have expectation equal to $1$ (the precise subset of these that have expectation $1$ is not important, since this does not affect the distribution of $\sum_i Y_i$). On the other hand, the number of collections $\{Y_i\}_i$ in sparse form is at most $(k^3+1)\cdot k^{3 k^2} \cdot (n+1) =  n \cdot \left({1 \over \epsilon}\right)^{O({1/\epsilon^2})}$, since there are $k^3+1$ choices for $\ell$, $k^{3k^2}$ choices for the expectations of variables $Y_1,\ldots,Y_{\ell}$ up to permutations of the indices of these variables (namely we need to choose how many of these $\ell$ variables have expectation $1/k^2$, how many have expectation $2/k^2$, etc.), and at most $n+1$ choices for the number of variables indexed by $i > \ell$ that have expectation equal to $1$. 

To improve on the size of the cover we show that we can remove from the aforementioned cover a large fraction of collections in sparse form. In particular, we shall only keep $n \cdot \left({1 \over \epsilon}\right)^{O(\log^2{1/\epsilon})}$ collections in sparse form, making use of Theorem~\ref{theorem:binomial appx theorem}. Indeed, consider a collection $\Y=\{Y_i\}_i$ in sparse form and let ${\cal L}_{\Y} =\{i~|~\E[Y_i] \in (0,1/2]\} \subseteq [n]$, ${\cal R}_{\Y}= \{i~|~\E[Y_i] \in (1/2,1)\} \subseteq [n]$. Theorem~\ref{theorem:binomial appx theorem} implies that, if we compare $\{Y_i\}_i$ with another collection $\{Z_i\}_i$ satisfying the following:
\begin{align}
\sum_{i \in {\cal L}} \E[Y_i]^{t} = \sum_{i \in {\cal L}} \E[Z_i]^{t},~~~\text{for all } t=1,\ldots,d;\\
\sum_{i \in {\cal R}} \E[Y_i]^{t} = \sum_{i \in {\cal R}} \E[Z_i]^{t},~~~\text{for all } t=1,\ldots,d;\\
\E[Y_i] = \E[Z_i],~~~\text{for all }i \in [n]\setminus ({\cal L} \cup {\cal R}),
\end{align}
then $d_{\rm TV}(\sum_i Y_i, \sum_i Z_i) \le 2\cdot 20(d+1)^{1/4} 2^{-(d+1)/2}$. In particular, for some $d(\epsilon)= O(\log 1/\epsilon)$, this bound becomes at most $\epsilon$. 

For a collection ${\cal Y}=\{Y_i\}_i$, we define the {\em moment profile $m_{\Y}$} of the collection to be a $(2 d(\epsilon)+1)$-dimensional vector 
$$m_{\Y} = \left(\sum_{i \in {\cal L}_{\Y}} \E[Y_i], \sum_{i \in {\cal L}_{\Y}} \E[Y_i]^{2},\ldots,\sum_{i \in {\cal L}_{\Y}} \E[Y_i]^{d(\epsilon)}; \sum_{i \in {\cal R}_{\Y}} \E[Y_i], \ldots,\sum_{i \in {\cal R}_{\Y}} \E[Y_i]^{d(\epsilon)} ; |\{i~|~\E[Y_i] = 1\}| \right).$$
By the previous discussion, for two collections $\Y=\{Y_i\}_i$ and $\Z=\{Z_i\}_i$, if $m_{\Y} = m_{\Z}$ then $d_{\rm TV}(\sum_i Y_i, \sum_i Z_i) \le \epsilon$.

Now given the $\epsilon$-cover produced in~\cite{D: anonymous 3 full} we perform the following sparsification operation: for every possible moment vector that can arise from a collection $\{Y_i\}_i$ in sparse form, we only keep in our cover one collection with such moment vector. The cover resulting from the sparsification operation is a $2 \epsilon$-cover, since the sparsification loses us an additive $\epsilon$ in total variation distance, as argued above. We now compute the size of the new cover. The total number of moment vectors arising from sparse-form collections of indicators is at most
$k^{O(d(\epsilon)^2)} \cdot (n+1)$.
Indeed, consider a collection $\Y$ in sparse form. There are at most  $k^3+1$ choices for $|{\cal L}_{\Y}|$, at most $k^3+1$ choices for $|{\cal R}_{\Y}|$, and at most $(n+1)$ choices for $|\{i~|~\E[Y_i] = 1\}|$. We claim next that the total number of possible vectors 
$$\left(\sum_{i \in {\cal L}_{\Y}} \E[Y_i], \sum_{i \in {\cal L}_{\Y}} \E[Y_i]^{2},\ldots,\sum_{i \in {\cal L}_{\Y}} \E[Y_i]^{d(\epsilon)}\right)$$
is at most $k^{O(d(\epsilon)^2)}$. Indeed, for all $t=1,\ldots,d(\epsilon)$, $\sum_{i \in {\cal L}_{\Y}} \E[Y_i]^{t} \le |{\cal L}_{\Y}|$ and it must be a multiple of $1/k^{2t}$. So the total number of possible values for $\sum_{i \in {\cal L}_{\Y}} \E[Y_i]^{t}$ is at most $(k^{2t} |{\cal L}_{\Y}| +1) \le (k^{2t} k^3 +1)$. It's easy to see then that the number of possible moment vectors $$\left(\sum_{i \in {\cal L}_{\Y}} \E[Y_i], \sum_{i \in {\cal L}_{\Y}} \E[Y_i]^{2},\ldots,\sum_{i \in {\cal L}_{\Y}} \E[Y_i]^{d(\epsilon)}\right)$$ is at most
$$\prod_{t=1}^{d(\epsilon)}(k^{2t} k^3 +1) \le k^{O(d(\epsilon)^2)}.$$
The same upper bound applies to the total number of possible moment vectors
$$\left(\sum_{i \in {\cal R}_{\Y}} \E[Y_i], \sum_{i \in {\cal R}_{\Y}} \E[Y_i]^{2},\ldots,\sum_{i \in {\cal R}_{\Y}} \E[Y_i]^{d(\epsilon)}\right).$$
It follows then that the total number of sparse-form collections of indicators that we have kept in our cover after the sparsification operation is at most $k^{O(d(\epsilon)^2)} \cdot (n+1) = n \cdot \left({1 \over \epsilon}\right)^{O(\log^2{1/\epsilon})}$. The number of collections in heavy Binomial form that we have in our cover is the same as before and hence it is at most $n^3 \cdot O(1/\epsilon)$. So the size of the sparsified cover is at most $n^3 \cdot O(1/\epsilon) + n \cdot \left({1 \over \epsilon}\right)^{O(\log^2{1/\epsilon})}$.

To finish the proof it remains to argue that we don't actually need to produce the cover of~\cite{D: anonymous 3 full} and subsequently sparsify it to obtain our cover, but we can produce it directly in time $n^3 \cdot O(1/\epsilon) + n \cdot \left({1 \over \epsilon}\right)^{O(\log^2{1/\epsilon})}$. We claim that given a moment vector $m$ we can compute a collection $\Y=\{Y_i\}_i$ such that $m_{\Y}=m$, if such a collection exists, in time $\left({1 \over \epsilon}\right)^{O(\log^2{1/\epsilon})}$. This follows from Claim~\ref{claim: easy to solve moment equations} in Appendix~\ref{appendix: analysis of moment search}.~\footnote{A naive application of Claim~\ref{claim: easy to solve moment equations} results in running time $n^3 \cdot \left({1 \over \epsilon}\right)^{O(\log^2{1/\epsilon})}$. However, we can proceed as follows: we can guess $|{\cal L}_{\Y}|$ and $|{\cal R}_{\Y}|$ (at most $O(k^6)$ guesses) and invoke Claim~\ref{claim: easy to solve moment equations} with $m_0=m_1=0$, $m=1$, $m_s=|{\cal L}_{\Y}|$ and $m_b=|{\cal R}_{\Y}|$ just to find $\{Y_i\}_{i \in {\cal L}_{\Y} \cup {\cal R}_{\Y}}$. To obtain the sought after $\Y=\{Y_i\}_i$ we then add $m_{2d(\epsilon)+1}$ indicators with expectation $1$ and make the remaining indicators $0$.} Hence, our algorithm producing the cover enumerates over all possible moment vectors and for each moment vector invokes Claim~\ref{claim: easy to solve moment equations} to find a consistent sparse collection of indicators, if such collection exists, adding that collection into the cover. Then it enumerates over collections of indicators in heavy Binomial form and adds them to the cover. The overall running time is as promised.
\end{prevproof}

\section{Discussion and Open Problems}
The mystery of PTAS for Nash equilibria deepens.  There are simple
algorithms for interesting special cases well within reach, and in
fact we have seen that the existence of a PTAS is not incompatible
with PPAD-completeness.  But oblivious algorithms cannot take us all
the way to the coveted PTAS for the general case.  In the important
special case of anonymous games, the approach of \cite{DP: anonymous
1, DP: anonymous 2, D: anonymous 3} --- by design involving oblivious
algorithms --- hits a brick wall of $({1\over
\epsilon})^{1\over \epsilon^\alpha}$, but then a more elaborate
probabilistic result about moments and Bernoulli sums breaks that
barrier.  Pseudopolynomial bounds, familiar from \cite{LMM},
show up in approximation algorithms for anonymous games as well.

Many open problems remain, of course:

\begin{itemize}

\item Is there a PTAS for Nash equilibria in general games?  A PTAS
for bimatrix games that exploits the linear programming-like nature
of the problem would not be unthinkable.

\item Find a truly practical, and hopefully evocative of strategically interacting
crowds, PTAS for anonymous games with two strategies.

\item Prove that finding an exact Nash equilibrium in an anonymous
game with a finite number of strategies is PPAD-complete.

\item Find a PTAS for 2-strategy graphical games --- the other
important class of multi-player games.

\item Alternatively, it is not unthinkable that the graphical
games special case above is PPAD-complete to approximate
sufficiently close.

\end{itemize}

\appendix
\arxiv{
\section{Bimatrix Games}
\subsection{PTAS for Small Probability Games} \label{appendix:linear support games}
\begin{prevproof}{Lemma}{lemma: restated lipton markakis mehta}
Let $X$ and $Y$ be independent random variables  such that
$$X = e_i, \text{with probability $x_i$, for all $i\in[n]$},$$
$$Y = e_j, \text{with probability $y_j$, for all $j\in[n]$}.$$
Let then $X_1,X_2,\ldots, X_t$  be $t$ copies of variable $X$ and
$Y_1$,$Y_2$, $\ldots$, $Y_t$  be $t$ copies of variable $Y$, where
the variables $X_1$,$\ldots$, $X_t$, $Y_1$, $\ldots$, $Y_t$ are mutually independent.
Setting $\X= \frac{1}{t}\sum_{k=1}^tX_k$ and $\mathcal{Y}=
\frac{1}{t}\sum_{k=1}^tY_k$, as in the statement of the theorem, we
will argue that, with high probability, $(\X,\mathcal{Y})$ is an
$\epsilon$-Nash equilibrium of $\G$.

Let $\U_i = e_i^{\text{\tiny T}} R \Y$ and $\V_i = \X^{\text{\tiny
T}} C e_i$ be the payoff of the row and column player respectively
for playing strategy $i \in [n]$. Fixing $i \in [n]$, we have that
$$\U_i = \frac{1}{t} \sum_{k=1}^t e_i^{\text{\tiny T}} R Y_k,$$
\begin{align*}\mathbb{E}[\U_i] &= \frac{1}{t}\sum_{k=1}^t\mathbb{E} [e_i^{\text{\tiny T}} R Y_k]= \frac{1}{t}\sum_{k=1}^t\mathbb{E} [e_i^{\text{\tiny T}} R Y]\\&=\mathbb{E} [e_i^{\text{\tiny T}} R Y]=\sum_{j\in [n]} y_j e_i^{\text{\tiny T}} R e_j= e_i^{\text{\tiny T}} R y.\end{align*}
An application of McDiarmid's inequality on the function
$f(\Lambda_1,\ldots,\Lambda_t)=\frac{1}{t}\sum_{k=1}^t \Lambda_k$,
where $\Lambda_k := e_i^{\text{\tiny T}} R Y_k$, $k \in [t]$, are
independent random variables, gives
$$\Pr[|\U_i - \mathbb{E}[\U_i]| \ge \epsilon / 2] \le 2e^{- \frac{t \epsilon^2}{8}} \le \frac{2}{n^2}.$$
Applying a union bound, it follows that with probability at least
$1-\frac{4}{n}$ the following properties are satisfied by the pair
of strategies $(\X,\Y)$:
\begin{align}
|e_i^{\text{\tiny T}} R \Y - e_i^{\text{\tiny T}} R y| \le \epsilon/2, \text{ for all $i\in[n]$};\\
|\X^{\text{\tiny T}} C e_j - x^{\text{\tiny T}} C e_j| \le
\epsilon/2, \text{ for all $j \in [n]$}.
\end{align}
The above imply that $(\X,\Y)$ is an $\epsilon$-Nash equilibrium.
Indeed, for all $i,i' \in [n]$, we have that
\begin{align*}
e_i^{\text{\tiny T}} R \Y &> e_{i'}^{\text{\tiny T}} R \Y+\epsilon\\ &\Rightarrow e_i^{\text{\tiny T}} R y \ge e_i^{\text{\tiny T}} R \Y -\epsilon/2 > e_{i'}^{\text{\tiny T}} R \Y+\epsilon/2 > e_{i'}^{\text{\tiny T}} R y\\
&\Rightarrow x_{i'} =0 \Rightarrow \X_{i'}=0,
\end{align*}
where the last implication follows from the fact that $\X$ is formed
by taking a distribution over samples from $x$. Similarly, it can be
argued that, for all $j,j' \in [n]$,
\begin{align*}
\X^{\text{\tiny T}} C e_j > \X^{\text{\tiny T}} C e_{j'}+\epsilon
\Rightarrow \Y_{j'}=0.
\end{align*}
\end{prevproof}

\begin{prevproof}{Claim}{claim:cardinality of good multisets is large}
We are only going to lower bound the cardinality of the set $\A$. A similar argument applies for $\B$. Recall that $(x,y)$ is a Nash equilibrium in which $x_i \le {1 \over \delta n}$, for all $i \in [n]$, and similarly for $y$. 

Let us now define the set $\A'$ as follows
$$\A':= \{A~\vline~ A \in \A \text{~and~}x_i>0, \forall i \in A\} \subseteq \mathcal{A},$$
that is, $\A'$ is the subset of $\A$ containing those multisets that could arise by taking $t$ samples from $x$. We are going to use the structure of the probability distribution $x$ and Lemma~\ref{lemma: restated lipton markakis mehta} to argue that $\A'$ is large. For this, let us consider the random experiment of taking $t$ independent samples from $x$ and forming the corresponding multiset $A$. Let also $\X$ be the uniform distribution over $A$. Lemma~\ref{lemma: restated lipton markakis mehta} asserts that with probability at least $1-{4 \over n}$ the distribution $\X$ will satisfy assertion~\ref{assertion 3 in LMM restated} of Lemma~\ref{lemma: restated lipton markakis mehta}. However, this does not directly imply a lower bound on the cardinality of $\A'$---it could be that all the probability mass is concentrated on a single element of $\A'$. However, the probability that a multiset arises by sampling $x$ is at most
$$\left({1 \over \delta n}\right)^{t},$$
since the probability mass that $x$ assigns to every $i \in [n]$ is at most ${{1 \over \delta n}}$. Therefore, the total number of distinct good multisets in $\A'$ is at least
$$\frac{1-{4 \over n}}{\left({1 \over \delta n}\right)^{t}}.$$
Therefore, $|\A'| \ge (1-{4 \over n}){\left({ \delta n}\right)^{t}}$ and $|\A| \ge |\A'| \ge(1-{4 \over n}){\left({ \delta n}\right)^{t}}$.
\end{prevproof}
\begin{prevproof}{Claim}{claim: probability of success of linear support large}
From Claim~\ref{claim:cardinality of good multisets is large}, it follows that, if a random multiset $A'$ is sampled, the probability that it is good is
\begin{align*}\Pr \left[ A' \in \A \right] \ge \frac{\left(1-\frac{4}{n}\right)\left( \delta n \right)^t}{n^t} &= \left(1- \frac{4}{n}\right) {\delta}^t\\ &=  \Omega \left(\delta \cdot n^{-16 \log(1/\delta)/{\epsilon^2}}\right),
\end{align*}
Similarly, if a random multiset $B'$ is sampled, the probability that it is good is
$$\Pr \left[ B' \in \B \right]  =  \Omega \left(\delta \cdot n^{-16 \log(1/\delta)/{\epsilon^2}}\right).$$
By independence, it follows that
$$\Pr \left[ A'\in \A \text{ and }B' \in \B \right]  =  \Omega \left(\delta^2 \cdot n^{-32 \log(1/\delta)/{\epsilon^2}}\right).$$ 
But, if $A'\in \A$, $B' \in \B$, then the uniform distribution $\X'$ over $A'$ and the uniform distribution $\Y'$ over $B'$ comprise an $\epsilon$-Nash equilibrium (see Proof of Lemma~\ref{lemma: restated lipton markakis mehta} for a justification). Hence,
\begin{align*}
&\Pr \left[ (\X',\Y') \text{ is an $\epsilon$-Nash equilibrium} \right]  \\&~~~~~~~~~~~~=  \Omega \left(\delta^2 \cdot n^{-32 \log(1/\delta)/{\epsilon^2}}\right).
\end{align*}
\end{prevproof}

\subsection{The Lower Bound for Bimatrix Games} \label{appendix:lower bound for bimatrix}
In Figure~\ref{fig: the payoff matrix RS}, we illustrate the game matrix $R_S$ in our construction in the Proof of Theorem~\ref{thm:no oblivious PTAS} (Section~\ref{sec:lower bound}), for the case $n=6$, $\ell=4$, $S=\{1,2,3,4\}$. Each column corresponds to a subset of $S$ of size $2$. We add then two extra rows to make the matrix square.
\begin{figure}[h]
$$\left(
\begin{tabular}{cccccc}
 1~~~~~    &1~~~~~     &1~~~~~     &0~~~~~   &0~~~~~     &0\\
 1~~~~~    &0~~~~~    &0~~~~~     &1~~~~~    &1~~~~~     &0\\
 0~~~~~    &1~~~~~     &0~~~~~     &1~~~~~    &0~~~~~     &1\\
 0~~~~~    &0~~~~~     &1~~~~~     &0~~~~~    &1~~~~~     &1\\
 \hdashline
-1~~~~~   &-1~~~~~    &-1~~~~~    &-1~~~~~   &-1~~~~~    &-1\\
-1~~~~~   &-1~~~~~    &-1~~~~~    &-1~~~~~   &-1~~~~~    &-1\\
\end{tabular}\right)$$
\caption{$R_S$ for the case $n=6$, $\ell=4$, $S=\{1,2,3,4\}$.} \label{fig: the payoff matrix RS}
\end{figure}

\begin{prevproof}{Lemma}{lem: claim about the properties of strategy of the row player at eps-equilibrium}
The first assertion is easy to justify. Indeed, no matter what $y$ is, the payoff $u_i$ of the row player for playing strategy $i$, $i \in [n]$, satisfies:
\begin{itemize}
\item $u_i = -1$, for all $i \notin S$;
\item $u_i \ge 0$, for all $i \in S$;
\end{itemize}
Hence, fixing any $i' \in S$, we have, for all $i \notin S$, that $u_{i'} > u_i + \epsilon$, which by the definition of an $\epsilon$-Nash equilibrium implies that $x_i=0$.

Towards justifying the second assertion, let $u$ be the utility of the row player at the $\epsilon$-Nash equilibrium $(x,y)$. Since, by the first assertion, $x_i=0$ for all $i \notin S$, it follows that the payoff of the column player is $1-u$, since the game restricted to the rows of the set $S$ is $1$-sum. Since every column of $R_S$ restricted to the rows of the set $S$ has exactly half of the entries equal to $1$ and the other half equal to $0$, it follows that
$$\sum_{i \in S}u_i = \sum_{j \in [n]}{\frac{\ell}{2}y_j} = \frac{\ell}{2}.$$
It follows that there exists some $i^* \in S$ such that $u_{i^*} \ge 1/2$. By the definition of $\epsilon$-Nash equilibrium, it follows then that
\begin{align}
u \ge u_{i^*} - \epsilon \ge 1/2 - \epsilon, \label{eq: bound on the payoff of the row player..}
\end{align}
otherwise the row player would be including in his support strategies which are more than $\epsilon$ worse than $i^*$.

Let us now consider the payoff of the column player for choosing various strategies $j \in [n]$. Without loss of generality let us assume that $S = \{1,2,\ldots, \ell \}$ and that $x_1 \ge x_2 \ge \ldots \ge x_{\ell}$. Let $j^*$ be such that $S_{j^*}=\{\ell/2+1, \ell/2+2,\ldots, \ell \}$. Then, by the definition of $C_S$, the payoff of the column player for choosing strategy $j^*$ is $$v_{j^*} = \sum_{i=1}^{\ell/2}x_i.$$
Since $(x,y)$ is an $\epsilon$-Nash equilibrium, it follows that $v_{j^*}$ should be within $\epsilon$ from the payoff of the column player. Hence,
$$v_{j^*} \le 1-u + \epsilon.$$
Combining the above with~\eqref{eq: bound on the payoff of the row player..}, it follows that
\begin{align*}
&v_{j^*} - \frac{1}{2}\le 2\epsilon ~~\Rightarrow~~ \sum_{i=1}^{\ell/2}\left(x_i - \frac{1}{\ell}\right) \le 2\epsilon. 
\end{align*}
The right hand side of the above, implies $\ell_1(x,u_S) \le 8\epsilon,$ by an application of Lemma~\ref{lem: bound on l1 from bound on the sum of the half of them} below, with $a_i = x_i - \frac{1}{\ell}$, for all $i \in [\ell]$, $k=2\epsilon$.

\begin{lemma} \label{lem: bound on l1 from bound on the sum of the half of them}
Let $\{a_i\}_{i=1}^{\ell}$ be real numbers satisfying the following properties for some $k \in \mathbb{R}_+$:
\begin{enumerate}
\item $a_1 \ge a_2 \ge \ldots \ge a_{\ell}$; \label{strange lemma:cond:1}

\item $\sum_{i=1}^{\ell}a_i =0$; \label{strange lemma:cond:2}

\item $\sum_{i=1}^{\ell/2}a_i \le k$. \label{strange lemma:cond:3}
\end{enumerate}
Then
\begin{align}
\sum_{i=1}^{\ell}|a_i| \le 4k. \label{eq:strange lemma:bound on l1 distance}
\end{align}
\end{lemma}

\begin{prevproof}{Lemma}{lem: bound on l1 from bound on the sum of the half of them}
We distinguish two cases: (i) $a_{\ell/2} \ge 0$ and (ii) $a_{\ell/2} < 0$. In the first case, we have from Conditions~\ref{strange lemma:cond:1} and~\ref{strange lemma:cond:3} that
$$\sum_{i=1}^{\ell /2}|a_i| \le k.$$
Using Condition~\ref{strange lemma:cond:1} some more we get
$$\sum_{i > \frac{\ell}{2}:~a_i \ge 0 }|a_i| \le \sum_{i > \frac{\ell}{2}:~a_i \ge 0 } a_{\ell/2} \le \frac{\ell}{2} a_{\ell/2} \le \sum_{i=1}^{\ell /2}|a_i| \le k.$$
Combining the above we get
\begin{align}
\sum_{i:~a_i \ge 0}|a_i| \le 2k. \label{proof of strange lemma: bound for the positive terms}
\end{align}
Now we employ Condition~\ref{strange lemma:cond:2}, to get
\begin{align}
\sum_{i:~a_i <0}|a_i| = \sum_{i:~a_i <0}(-a_i) = \sum_{i:~a_i \ge 0} a_i =  \sum_{i:~a_i \ge 0} |a_i| \le 2k. \label{proof of strange lemma: bound for the negative terms}
\end{align}
We combine~\eqref{proof of strange lemma: bound for the positive terms},~\eqref{proof of strange lemma: bound for the negative terms} to deduce~\eqref{eq:strange lemma:bound on l1 distance}.
Case (ii) is treated by repeating the above argument with
$$a_i \leftarrow (- a_{\ell-i+1})\text{, for all $i \in [\ell]$}.$$
\end{prevproof}\end{prevproof}

\begin{prevproof}{Claim}{claim:superpolynomially many disjoint balls}  Consider the set $V:=\{u_S| S \subseteq [n],~|S|=\ell\}$; note that $|V| = \Omega(n^{.8 \log_2 n})$.  For every $u_S\in V$ consider the set $N(u_S)$ of all other $u_{S'}$'s that are within $\ell_1$ distance $17\epsilon$ from $u_S$;  it is easy to see that $|N(u_S)| = O( n^{34\epsilon \log_2 n})$.  Therefore, we can select a subset $V' \subseteq V$ of at least ${|V| \over \max_S|N(u_S)|}=\Omega\left(
n^{(.8-34\epsilon)\log_2 n}\right)$ elements of $V$ so that every pair of elements is at $\ell_1$ distance at least $17\epsilon$ apart; it follows that for every pair of elements in $V'$ their $\ell_1$ balls of radius $8\epsilon$ are disjoint.  The proof is complete.
\end{prevproof}}

{
\section{The Lower Bound for Oblivious PTAS's for Anonymous Games} \label{appendix:lower bound for anonymous}
\subsection{Constructing Anonymous Games with Prescribed Equilibria} \label{appendix:lower bound for anonymous:games with unique equilibria}

\begin{theorem} \label{theorem: anonymous games with unique equilibria up to plus minus small}
For any collection $\P:=(p_i)_{i \in [k]}$, where $p_i \in [3 \delta
k,1]$, for some $\delta>0$ and $k\in \mathbb{N}$, there exists an
anonymous game $\G_\P$ with $k+2$ players, $2$ strategies, $0$ and
$1$, payoffs in $[-1,1]$, and three player types A,B,C, in which $k$
players, $1,\ldots,k$, belong to type A, $1$ player belongs to type
B, and $1$ player belongs to type C, and such that in every
$\delta'$-Nash equilibrium, where $\delta' < \delta$, the following
is satisfied: For every $i$, player $i$'s mixed strategy belongs to
the set $[p_i - 7 k^2\delta,p_i+7 k^2 \delta]$; moreover, at least
one of the players belonging to types B and C play strategy $1$ with
probability $0$.
\end{theorem}

\begin{proof}
Let us call $B$ the single player of type $B$ and $C$ the single
player of type $C$. Let us also use the notation: $\mu= \sum_{i \in
[k]} p_i$, and $\mu_{-i}= \sum_{j \in [k] \setminus \{i\}} p_j$, for
all $i$. Now, let us assign the following payoffs to the players $B$
and $C$:
\begin{itemize}
\item $u^B_1 = {1 \over k} \cdot (t_A - \mu)$, where $t_A$ is the number of players of type A who play strategy $1$;
\item $u^B_0 = 2 \delta$;
\item $u^C_1 = {1 \over k} \cdot (\mu - t_A)$, where $t_A$ is the number of players of type A who play strategy $1$;
\item $u^C_0 = 2 \delta$;
\end{itemize}
The payoff functions of the players of type A are defined as
follows. For all $i \in [k]$:
\begin{itemize}
\item $u^i_0 = {1 \over k} (\mu_{-i} \cdot \X_{\text{B plays $0$}} \cdot \X_{\text{C plays $0$}}  - \delta k \cdot \X_{\text{C plays $1$}})$, where $\X_{\text{B plays $0$}}$, $\X_{\text{C plays $0$}}$ and $\X_{\text{C plays $1$}}$ are the indicators of the events `\text{B plays $0$}', `\text{C plays $0$}' and `\text{C plays $1$}' respectively.
\item $u^i_1 = {1 \over k} (t_{A,-i} \cdot \X_{\text{B plays $0$}} \cdot \X_{\text{C plays $0$}}  - \delta k \cdot \X_{\text{B plays $1$}})$, where $t_{A,-i}$ is the number of players of type A who are different than $i$ and play $1$, and $\X_{\text{B plays $0$}}$, $\X_{\text{C plays $0$}}$ and $\X_{\text{B plays $1$}}$ are the indicators of the events `\text{B plays $0$}', `\text{C plays $0$}' and `\text{B plays $1$}' respectively.
\end{itemize}
Note that the range of all payoff functions of the game thus defined
is $[-1,1]$. We now claim the following:
\begin{claim}\label{claim: lower bound: mu primes are close to mu's}
In every $\delta'$-Nash equilibrium, where $\delta' < \delta$, it
must be that
$$\sum_{i \in [k]} q_i = \mu \pm 3 \delta k,$$
where $q_1,\ldots,q_k$ are the mixed strategies of the players
$1,\ldots,k$.
\end{claim}
\begin{prevproof}{Claim}{claim: lower bound: mu primes are close to mu's}
Let $\mu' = \sum_{i \in [k]}q_i$. Suppose for a contradiction that
in a $\delta'$-Nash equilibrium $\mu' > \mu + 3 \delta k$; then
$${1 \over k} (\mu' - \mu) > 3 \delta.$$
Note however that $\E[u^B_1]= {1 \over k} (\mu' - \mu)$ and
$\E[u^C_1]= - {1 \over k} (\mu' - \mu)$. Hence, the above implies
\begin{align}
\E[u^B_1] > \E[u^B_0] + \delta,\\
\E[u^C_1] < \E[u^C_0] - \delta.
\end{align}
Since $\delta' <\delta$, it must be then that $\Pr[\text{B plays
$1$}]=1$ and $\Pr[\text{C plays $1$}]=0$. It follows then that for
all $i \in [k]$:
\begin{align*}
\E[u^i_0] = 0,\\
\E[u^i_1] =  -\delta.
\end{align*}
Hence, in a $\delta'$-Nash equilibrium with $\delta' <\delta$, it must
be that $\Pr[\text{$i$ plays $1$}]=q_i=0$, for all $i\in[k]$. This
is a contradiction since we assumed that $\mu'=\sum_{i \in [k]}q_i >
\mu + 3\delta k$, and $\mu$ is non-negative. Via similar arguments
we show that the assumption $\mu' < \mu - 3 \delta k$ also leads to
a contradiction. Hence, in every $\delta'$-Nash equilibrium with
$\delta' < \delta$, it must be that
$$\mu' = \mu \pm 3\delta k.$$
\end{prevproof}

We next show that in every $\delta'$-Nash equilibrium with $\delta' <
\delta$, at least one of the players $B$ and $C$ will not include
strategy $1$ in her support.
\begin{claim}\label{claim: lower bound: players B and C cannot be both supported on strategy 1}
In every $\delta'$-Nash equilibrium with $\delta' < \delta$,
$$\Pr[\text{B plays $1$}]=0 \text{~or~}\Pr[\text{C plays $1$}]=0.$$
\end{claim}
\begin{prevproof}{Claim}{claim: lower bound: players B and C cannot be both supported on strategy 1}
Let $q_1,\ldots,q_k$ be the mixed strategies of players $1,\ldots,k$
at some $\delta'$-Nash equilibrium of the game with $\delta' <
\delta$. Let us consider the quantity $\M= {1 \over k} (\mu' -
\mu)$, where $\mu'= \sum_{i \in [k]}q_i$. We distinguish the
following cases:
\begin{itemize}
\item $\M \le \delta$: In this case, $\E[u^B_1]= {1 \over k} (\mu' - \mu) \le \delta \le 2\delta - \delta=\E[u^B_0]-\delta$. Since $\delta'<\delta$, $\Pr[\text{B plays $1$}]=0$.
\item $\M \ge \delta$: In this case, $\E[u^C_1]= - {1 \over k} (\mu' - \mu) \le - \delta \le 2\delta - \delta=\E[u^C_0]-\delta$. Since $\delta'<\delta$, $\Pr[\text{C plays $1$}]=0$.
\end{itemize}
\end{prevproof}

Finally, we establish the following.
\begin{claim}\label{claim: lower bound: mu -i primes are close to mu -i's}
In every $\delta'$-Nash equilibrium with $\delta' < \delta$, it
must be that for every player $i \in [k]$:
$$\mu'_{-i}:=\sum_{j \in [k] \setminus \{i\}} q_j = \mu_{-i} \pm 4 \delta k^2,$$
where $q_1,\ldots,q_k$ are the mixed strategies of the players
$1,\ldots,k$.
\end{claim}
\begin{prevproof}{Claim}{claim: lower bound: mu -i primes are close to mu -i's}
Let us fix any $\delta'$-Nash equilibrium. From Claim~\ref{claim:
lower bound: players B and C cannot be both supported on strategy 1}
it follows that either player B or C plays strategy $1$ with
probability $0$. Without loss of generality, we will assume that
$\Pr[\text{C plays $1$}]=0$ (the argument for the case $\Pr[\text{B
plays $1$}]=0$ is identical to the one that follows).

Let us now fix a player $i \in [k]$. We show first that under the
assumption $\Pr[\text{C plays $1$}]=0$, $\Pr[\text{C plays $0$}]=1$,
it must be that
\begin{align}
\mu_{-i} \le \mu'_{-i} + \delta k. \label{eq:lower bound mu -i's eq
1}
\end{align}
Assume for a contradiction that $\mu_{-i} > \mu'_{-i} + \delta k$.
It follows then that
\begin{align*}
&\mu_{-i} (1-\Pr[\text{B plays $1$}]) \ge \mu'_{-i} (1-\Pr[\text{B plays $1$}])\\&~~~~~~~~~~~~~~~~~~~~~~~~~~~~~~+ \delta k (1-\Pr[\text{B plays $1$}])\\
&\Rightarrow~~\mu_{-i} \Pr[\text{B plays $0$}] \ge \mu'_{-i} \Pr[\text{B plays $0$}]\\ &~~~~~~~~~~~~~~~~~~~~~~~~~~~~~~+ \delta k (1-\Pr[\text{B plays $1$}])\\
&\Rightarrow~~\mu_{-i} \Pr[\text{B plays $0$}] \Pr[\text{C plays $0$}] \ge \\&~~~~~~~~~~~~~~~~~~~~~~~~~~~~ \mu'_{-i} \Pr[\text{B plays $0$}] \Pr[\text{C plays $0$}] \\&~~~~~~~~~~~~~~~~~~~~~~~~~~~~~~~~- \delta k\Pr[\text{B plays $1$}] +\delta k\\
&\Rightarrow~~\E[u^i_0] \ge \E[u^i_1]  +\delta.
\end{align*}
But we assumed that we fixed a $\delta'$-Nash equilibrium with
$\delta'<\delta$; hence the last equation implies that $q_i=0$. But
this leads quickly to a contradiction since, if $q_i=0$, then using
Claim~\ref{claim: lower bound: mu primes are close to mu's} we have
$$\mu'_{-i}=\mu' \ge \mu - 3 \delta k \ge \mu - p_i = \mu_{-i},$$
where we also used that $p_i \ge 3 \delta k$. The above inequality
contradicts our assumption that $\mu_{-i} > \mu'_{-i} + \delta k$.
Hence, \eqref{eq:lower bound mu -i's eq 1} must be satisfied. Using
that $\mu' \le \mu + 3\delta k$, which is implied by
Claim~\ref{claim: lower bound: mu primes are close to mu's}) we get
$$q_i \le p_i + 4\delta k.$$

From the above discussion it follows that
\begin{align}
q_j \le p_j + 4\delta k, \text{for all $j$}.\label{eq:lower bound mu
-i's eq 2}
\end{align}
Now fix $i \in [k]$ again. Summing~\eqref{eq:lower bound mu -i's eq
2} over all $j \neq i$, we get that
\begin{align}
\mu'_{-i} \le \mu_{-i} + 4\delta k^2. \label{eq:lower bound mu -i's
eq 3}
\end{align}
Combining \eqref{eq:lower bound mu -i's eq 1} and \eqref{eq:lower
bound mu -i's eq 3} we get
$$\mu'_{-i} = \mu_{-i} \pm 4\delta k^2.$$
\end{prevproof}
To conclude the proof of Theorem~\ref{theorem: anonymous games with
unique equilibria up to plus minus small}, we combine
Claims~\ref{claim: lower bound: mu primes are close to mu's} and
\ref{claim: lower bound: mu -i primes are close to mu -i's}, as
follows. For every player $i \in [k]$, we have from
Claims~\ref{claim: lower bound: mu primes are close to mu's} and
\ref{claim: lower bound: mu -i primes are close to mu -i's} that in
every $\delta'$-Nash equilibrium with $\delta' < \delta$,
$$\mu'_{-i} = \mu_{-i} \pm 4\delta k^2~\text{and}~\mu' = \mu \pm 3 \delta k.$$
By combining these equations we get
$$q_i=p_i \pm 7 \delta k^2.$$
\end{proof}

\subsection{The Lower Bound}  \label{appendix:lower bound for anonymous:proof of the lower bound}

Given Theorem~\ref{theorem: anonymous games with unique equilibria up to plus minus small}, we can establish our lower bound.

\begin{prevproof}{Theorem}{theorem: no subexponential oblivious
PTAS} Let us fix any oblivious $\epsilon$-approximation algorithm for anonymous
games with $2$-strategies and $3$-player types. The algorithm comes
together with a distribution over unordered sets of mixed strategies---parametrized by the number of players $n$---which we denote by $D_{n}$.
%

We will consider the performance of the algorithm on the family of games specified in the statement of
Theorem~\ref{theorem: anonymous games with unique equilibria up to
plus minus small} for the following setting of parameters:
\begin{align*}
&k= \lfloor (1/ \epsilon)^{1/3} \rfloor,~\delta=1.01 \epsilon,~\P \in \mathcal{T}_{\epsilon}^{k}\\
\text{where}~&\mathcal{T}_{\epsilon}:=\left\{ j\cdot15
\epsilon^{1/3}~\vline~j = 1, \ldots, t_{\epsilon}
\right\},~~t_{\epsilon}=\left \lfloor {1 \over 15} \epsilon^{-1/3}
\right \rfloor.
\end{align*}
For technical reasons, let us define the following notion of distance between $\P, \Q
\in \mathcal{T}_{\epsilon}^k$.
$$d(\P,\Q) := \sum_{j=1}^{t_{\epsilon}} \Big| v^\P_j- v^\Q_j\Big|.$$
where $v^{\P}=(v^{\P}_1, v^{\P}_2,\ldots,v^{\P}_{t_\epsilon})$ is a
vector storing the frequencies of various elements of the set
$\mathcal{T}_{\epsilon}$ in the collection $\P$, i.e.
$v^{\P}_j:=|\{i~\vline~i\in[k], p_i = j\cdot15 \epsilon^{1/3}\}|$.
To find the distance between two collections $\P, \Q$ we compute the
$\ell_1$ distance of their frequency vectors. Notice in particular that this distance must be an even number. We also need the following definition.
\begin{definition}
We say that two anonymous games $\G$ and $\G'$ share an
$\epsilon$-Nash equilibrium in unordered form if there exists an
$\epsilon$-Nash equilibrium $\sigma_{\G}$ of game $\G$ and an
$\epsilon$-Nash equilibrium $\sigma_{\G'}$ of game $\G'$ such that
$\sigma_{\G}$ and $\sigma_{\G'}$ are equal as unordered sets of
mixed strategies.
\end{definition}
We show first the following about the shareability of $\epsilon$-Nash equilibria among the games $\G_\P$, $\P \in T_{\epsilon}^k$.
\begin{claim}\label{claim:lower bound: no two games share an equilibrium}
If, for $\P, \Q \in T_{\epsilon}^k$, $d(\P,\Q) >0$, then there is no
$\epsilon$-Nash equilibrium that is shared between the games $\G_\P$
and $\G_\Q$ in unordered form.
\end{claim}
\begin{prevproof}{Claim}{claim:lower bound: no two games share an equilibrium}
For all $j$, let us define the $7.07k^2\epsilon$ ball around
probability $j\cdot15 \epsilon^{1/3}$ in the natural way:
$$B_j:=[j\cdot15 \epsilon^{1/3}-7.07k^2\epsilon, j\cdot15 \epsilon^{1/3}+7.07k^2\epsilon].$$
Observe that for all $j \ge 2$:
$$(j+1)\cdot15 \epsilon^{1/3} - j\cdot15 \epsilon^{1/3}=15 \epsilon^{1/3} > 2 \cdot 7.07k^2\epsilon.$$
Hence, for all $j, j'$: $B_j \cap B_j' = \emptyset$.

Now, let us consider any pair of $\epsilon$-Nash equilibria
$\sigma_{\G_\P}$, $\sigma_{\G_\Q}$ of the games $\G_\P$ and $\G_\Q$
and let us consider the vectors
$v^{\sigma_{\G_\P}}=(v^{\sigma_{\G_\P}}_1,\ldots,v^{\sigma_{\G_\P}}_{t_{\epsilon}})$
and
$v^{\sigma_{\G_\Q}}=(v^{\sigma_{\G_\Q}}_1,\ldots,v^{\sigma_{\G_\Q}}_{t_{\epsilon}})$
whose $j$-th components are defined as follows:
$$v^{\sigma_{\G_\P}}_j = \left( \begin{minipage}{5.5cm}\centering number of players who are assigned a mixed strategy from the set $B_j$ in $\sigma_{\G_\P}$ \end{minipage}  \right),$$
$$v^{\sigma_{\G_\Q}}_j = \left( \begin{minipage}{5.5cm}\centering number of players who are assigned a mixed strategy from the set $B_j$ in $\sigma_{\G_\Q}$ \end{minipage}  \right).$$
It is not hard to see that Theorem~\ref{theorem: anonymous games
with unique equilibria up to plus minus small} and our assumption
$d(\P,\Q)>0$ imply that $\| v^{\sigma_{\G_\P}} - v^{\sigma_{\G_\Q}}
\|_1 > 0,$ hence $\sigma_{\G_\P}$ and $\sigma_{\G_\Q}$ cannot be
permutations of each other. This concludes the proof.
\end{prevproof}

Next, we show that there exists a large family of games such that no two members of the family share an $\epsilon$-Nash equilibrium.
\begin{claim}\label{claim:lower bound:many sets}
There exists a subset $T \subseteq \mathcal{T}^k_{\epsilon}$ such
that:
\begin{enumerate}
\item for every $\P, \Q \in T$: $d(\P,\Q) >0$; \label{item in lower bound construction:property 1}
\item $|T| \ge 2^{\Omega \left(\left({1 \over \epsilon}\right)^{1/3}\right)}$; \label{item in lower bound construction:property 2}
\end{enumerate}
\end{claim}
\begin{prevproof}{claim}{claim:lower bound:many sets}
The total number of distinct multi-sets of cardinality $k$ with
elements from $\mathcal{T}_\epsilon$ is
$${t_\epsilon+k-1 \choose k}.$$
Hence, it is easy to create a subset $T \subseteq
\mathcal{T}^k_{\epsilon}$ such that:
\begin{itemize}
\item for every $\P, \Q \in T$: $d(\P,\Q) > 0$;
\item $|T| = {t_\epsilon+k-1 \choose k}$.
\end{itemize}

\noindent Clearly, the set $T$ satisfies Property
\ref{item in lower bound construction:property 1} in the statement.
For the cardinality bound we have:
\begin{align*}
|T| \ge {t_\epsilon+k-1 \choose k} &\ge \left(t_\epsilon+k-1 \over k\right)^k \\
&\ge \left(1+{1 \over 15} - {2 \over k}\right)^k  \ge 2^{\Omega \left(\left({1 \over
\epsilon}\right)^{1/3}\right)}.
\end{align*}
\end{prevproof}
Now let us consider the performance of the distribution $D_{k}$ on the family of anonymous games $\{\G_\P\}_{\P \in
T}$, where $T$ is the set defined in Claim~\ref{claim:lower
bound:many sets}. By Claims~\ref{claim:lower bound: no two games
share an equilibrium} and~\ref{claim:lower bound:many sets}, no two games in the
family share an $\epsilon$-Nash equilibrium in
unordered form. Hence, no matter what $D_{k}$ is, there will be some game in our family for which the probability that $D_{k}$ samples an $\epsilon$-Nash equilibrium of that game is at most
$$1/|T| \le 2^{-\Omega \left(\left({1 \over \epsilon}\right)^{1/3}\right)}.$$
This concludes the proof of Theorem~\ref{theorem: no subexponential oblivious PTAS}.
\end{prevproof}
}
\arxiv{
\section{The Binomial Approximation Theorem}\label{appendix:binomial approximation theorem}

\begin{prop}\label{proposition:variable moments to probability moments}
Condition $(C_d)$ in the statement of Theorem \ref{theorem:binomial
appx theorem} is equivalent to the following condition:

$$(V_d):~~\E\left[\left(\sum_{i=1}^n X_i\right)^{\ell}\right] = \E\left[\left(\sum_{i=1}^n Y_i\right)^{\ell}\right],~\text{for all } \ell \in [d].$$
\end{prop}
\begin{prevproof}{Proposition}{proposition:variable moments to probability moments}\text{}\\
$(V_d) \Rightarrow (C_d)$:~It is not hard to see that
$\E\left[\left(\sum_{i=1}^n X_i\right)^{\ell}\right]$ can be written
as a weighted sum of the {\em elementary symmetric polynomials}
$\psi_1(\P)$, $\psi_2(\P)$,...,$\psi_{\ell}(\P)$ with positive
coefficients, where, for all $\ell$, $\psi_{\ell}(\P)$ is defined as
$$\psi_{\ell}(\P):=\sum_{\begin{minipage}{1.5cm}\centering$S \subseteq [n]$\\$|S|=\ell$\end{minipage}} \prod_{i \in S} p_i.$$
$(V_d)$ then implies by induction
\begin{align}\psi_{\ell}(\P)=\psi_{\ell}(\Q),~~~\text{for all $\ell = 1,\ldots,d$}.\label{eq: equality of the elementary symmetric polynomials}
\end{align}
Now, for all $\ell$, define $\pi_{\ell}(\P)$ as the power sum
symmetric polynomial of degree $\ell$
$$\pi_{\ell}(\mathcal{P}):=\sum_{i=1}^np_i^{\ell}.$$
Now fix any $\ell \le d$. Observe that $\pi_{\ell}(\P)$ is a
symmetric polynomial of degree $\ell$ on the variables
$p_1,\ldots,p_n$. It follows (see, e.g.,~\cite{Zolotarev}) that
$\pi_{\ell}(\mathcal{P})$ can be expressed as a function of the
elementary symmetric polynomials
$\psi_1(\P),\ldots,\psi_{\ell}(\P)$. Since, by~\eqref{eq: equality
of the elementary symmetric polynomials}, $\psi_j(\P) = \psi_j(\Q)$,
for all $j \le \ell$, we get that $\pi_{\ell}(\mathcal{P}) =
\pi_{\ell}(\mathcal{Q})$. Since this holds for any $\ell \le d$,
$(C_d)$ is satisfied.

The implication $(C_d) \Rightarrow (V_d)$ is established in a
similar fashion. $(C_d)$ implies
$$\pi_{\ell}(\mathcal{P})=\pi_{\ell}(\mathcal{Q}), \text{for all $\ell=1,\ldots,d$}.$$
For any $\ell \le d$, $\E\left[\left(\sum_{i=1}^n
X_i\right)^{\ell}\right]$ is a symmetric polynomial of degree $\ell$
on the variables $p_1,\ldots,p_n$. It follows (see,
e.g.,~\cite{Zolotarev}) that $\E\left[\left(\sum_{i=1}^n
X_i\right)^{\ell}\right]$ can be expressed as a function of the
power-sum symmetric polynomials $\pi_1(\P),\ldots,\pi_{\ell}(\P)$.
And since $\pi_j(\P) = \pi_j(\Q)$, for all $j \le \ell$, we get
$\E\left[\left(\sum_{i=1}^n
X_i\right)^{\ell}\right]=\E\left[\left(\sum_{i=1}^n
Y_i\right)^{\ell}\right]$. Since this holds for any $\ell \le d$,
$(V_d)$ is satisfied.
%
\end{prevproof}

\begin{prevproof}{Theorem}{theorem:binomial appx theorem}
Let $X=\sum_i X_i$. The following theorem due to Roos~\cite{Roos},
specifies an expansion of the distribution function of $X$ as a sum
of a finite number of signed measures: the binomial distribution
$\mathcal{B}_{n,p}(m)$ (for an arbitrary choice of $p$) and its
first  $n$ derivatives with respect to the parameter $p$, at the
chosen value of $p$. More precisely,

\begin{theorem}[\cite{Roos}] \label{thm:roos}
Let $\mathcal{P}:=(p_i )_{i=1}^n$ be an arbitrary set of probability
values in $[0,1]$ and $\mathcal{X}:=(X_i)_{i=1}^n$ a collection of
independent indicators with $\E[X_i]=p_i$, for all $i\in [n]$; also
let $X:=\sum_i X_i$. Then, for all $m \in \{0,\ldots,n\}$ and any $p
\in (0,1)$,
\begin{align}
Pr[X = m]  = \sum_{\ell = 0}^n \alpha_{\ell}(\mathcal{P}, p)\cdot
\delta^{\ell}\mathcal{B}_{n,p}(m), \label{eq:krawtchouk expansion}
\end{align}
where in the above $\alpha_0(\mathcal{P},p):=1$,
\begin{align}
\alpha_{\ell}(\mathcal{P},p):= \sum_{1 \le k(1) < \ldots < k(\ell)
\le n} \prod_{r=1}^{\ell}(p_{k(r)}-p),~~\text{for all }\ell \in
[n], \label{eq: krawtchouk coefficients}
\end{align}
and
$$\delta^{\ell}\mathcal{B}_{n,p}(m):=\frac{(n-\ell)!}{n!} \frac{\partial^{\ell}}{\partial p^{\ell}}\mathcal{B}_{n,p}(m),$$
and for the purposes of the last definition we interpret
$$\mathcal{B}_{n,p}(m):=b(m,n,p)$$ as a function of the arguments
$m,n,p$ in the natural way:
$$b(m,n,p):=
\begin{cases}
{n \choose m} p^m (1-p)^{n-m},~\text{for}~n,m \in Z_{+},~m \le n;\\
~~~0,~~~~~~~~~~~~~~~~~\text{otherwise}.
\end{cases}$$
\end{theorem}

Assume now that we are given two collections $\X$, $\Y$ of
indicators as in the statement of Theorem~\ref{theorem:binomial appx
theorem}. We claim the following
\begin{lemma}\label{lemma:agreeing coefficients}
For sets $\P$ and $\Q$ satisfying property $(C_d)$, and for all $p$:
$$\alpha_{\ell}(\mathcal{P},p) = \alpha_{\ell}(\mathcal{Q},p),~~~\text{for all $\ell=0,\ldots,d$}.$$
\end{lemma}
\begin{prevproof}{lemma}{lemma:agreeing coefficients}
Clearly, $\alpha_{0}(\mathcal{P},p) = \alpha_{0}(\mathcal{Q},p)$.
Let us fix $\ell \in \{1,\ldots,d\}$. Observe that
$\alpha_{\ell}(\mathcal{P},p)$ is a symmetric polynomial of degree
$\ell$ on the variables $p_1,\ldots,p_n$. It follows (see,
e.g.,~\cite{Zolotarev}) that $\alpha_{\ell}(\mathcal{P},p)$ can be
expressed as a function of the power-sum symmetric polynomials
$\pi_1(\P),\ldots,\pi_{\ell}(\P)$ defined as
$$\pi_j(\mathcal{P}):=\sum_{i=1}^np_i^j,~\text{for all $j \in [\ell]$}.$$
It follows from $(C_d)$ that $\pi_j(\P) = \pi_j(\Q)$, for all $j \le
\ell$; from the previous discussion, this implies that
$\alpha_{\ell}(\mathcal{P},p) = \alpha_{\ell}(\mathcal{Q},p)$.
\end{prevproof}

For any $p\in(0,1)$, by combining Lemma~\ref{lemma:agreeing
coefficients} and Theorem~\ref{thm:roos} we get
that

\begin{align*}&Pr[X = m] - Pr[Y=m] \\&~~~~= \sum_{\ell = d+1}^n (\alpha_{\ell}(\mathcal{P}, p)-\alpha_{\ell}(\mathcal{Q}, p))\cdot \delta^{\ell}\mathcal{B}_{n,p}(m),~\text{ for all $m$}.
\end{align*}
Therefore, for all $p$:
\begin{align}
||X; Y|| &\le \frac{1}{2} \sum_{m=0}^n|Pr[X = m] - Pr[Y=m]|\\
&\le\frac{1}{2} \sum_{\ell = d+1}^n |\alpha_{\ell}(\mathcal{P}, p)-\alpha_{\ell}(\mathcal{Q}, p)|\cdot \| \delta^{\ell}\mathcal{B}_{n,p}(\cdot)\|_1\\
&\le\frac{1}{2} \sum_{\ell = d+1}^n
\left(|\alpha_{\ell}(\mathcal{P}, p)|+|\alpha_{\ell}(\mathcal{Q},
p)|\right)\cdot \| \delta^{\ell}\mathcal{B}_{n,p}(\cdot)\|_1.
\label{eq:tv in terms of remaining terms coefficients}
\end{align}
From~Theorem 2 in~\cite{Roos}, it follows that
\begin{align*}
&\frac{1}{2} \sum_{\ell = d+1}^n |\alpha_{\ell}(\mathcal{P}, p)|\cdot
\| \delta^{\ell}\mathcal{B}_{n,p}(\cdot)\|_1~~~~~~~~~~~~~~~~~~~~~~~~~~~~~\\&~~~~ \le
\frac{\sqrt{e}(d+1)^{1/4}}{2} \theta(\P,p)^{(d+1)/2} \frac{1-
\frac{d}{d+1}\sqrt{\theta(\P,p)}}{(1-\sqrt{\theta(\P,p)})^2},
\end{align*}
where
$$\theta(\P,p):= \frac{2 \sum_{i=1}^n(p_i - p)^2 + (\sum_{i=1}^n(p_i - p))^2}{2np(1-p)}$$
Choosing $p = \bar{p}:= \frac{1}{n}\sum_{i} p_i$, we get
(see~\cite{Roos})
$$\theta(\P,\bar{p})= \frac{\sum_{i=1}^n(p_i - \bar{p})^2 }{n\bar{p}(1-\bar{p})} \le |p_{\max}-p_{\min}| \le \frac{1}{2},$$
where $p_{\max}= \max_i{\{p_i\}}$ and $p_{\min}= \min_i{\{p_i\}}$.
From the above we get
\begin{align*}
&\frac{1}{2} \sum_{\ell = d+1}^n |\alpha_{\ell}(\mathcal{P},
\bar{p})|\cdot \| \delta^{\ell}\mathcal{B}_{n,\bar{p}}(\cdot)\|_1
\\&~~~~~~~~~~~\le \sqrt{e}(d+1)^{1/4} 2^{-(d+1)/2} \frac{1-
\frac{1}{\sqrt{2}}\frac{d}{d+1}}{(\sqrt{2}-1)^2}\\&~~~~~~~~~~~ \le 10(d+1)^{1/4}
2^{-(d+1)/2}.
\end{align*}
Since from $(C_d)$ we have that $\bar{p}=\sum_i p_i = \sum_i q_i$,
we get in a similar fashion
\begin{align*}
\frac{1}{2} \sum_{\ell = d+1}^n |\alpha_{\ell}(\mathcal{Q},
\bar{p})|\cdot \| \delta^{\ell}\mathcal{B}_{n,\bar{p}}(\cdot)\|_1
\le 10(d+1)^{1/4} 2^{-(d+1)/2}.
\end{align*}
Plugging these bounds into~\eqref{eq:tv in terms of remaining terms
coefficients} we get
$$||X; Y|| \le 20(d+1)^{1/4} 2^{-(d+1)/2}.$$
\end{prevproof}




\begin{corollary} \label{theorem:binomial appx theorem for large guys}
Let $\mathcal{P}:=(p_i )_{i=1}^n \in [1/2,1)^n$ and $\mathcal{Q}:=(q_i)_{i=1}^n \in [1/2,1)^n$ be two collections of  probability values in $[1/2,1)$. Let also $\mathcal{X}:=(X_i)_{i=1}^n$ and $\mathcal{Y}:=(Y_i)_{i=1}^n$ be two collections of independent indicators with $\E[X_i]=p_i$ and $\E[Y_i]=q_i$, for all $i \in [n]$. If for some $d \in [n]$ Condition $(C_d)$ is satisfied, then
$$\left|\left|\sum_{i}{X_i}~;~\sum_{i}{Y_i} \right|\right|= 20(d+1)^{1/4} 2^{-(d+1)/2}.$$
\end{corollary}

\begin{prevproof}{Corollary}{theorem:binomial appx theorem for large guys}
Define $X'_i = 1-X_i$, and $Y'_i=1-Y_i$, for all $i$. Apply
Theorem~\ref{theorem:binomial appx theorem} to deduce
$$\left|\left|\sum_{i}{X'_i}~;~\sum_{i}{Y'_i} \right|\right|= 20(d+1)^{1/4} 2^{-(d+1)/2}.$$
The proof is completed by noting
$$\left|\left|\sum_{i}{X_i}~;~\sum_{i}{Y_i} \right|\right|=\left|\left|\sum_{i}{X'_i}~;~\sum_{i}{Y'_i} \right|\right|.$$
\end{prevproof}

\section{The non-oblivious PTAS for Anonymous Games}

\subsection{The Oblivious PTAS of~~[Daskalakis, 2008]} \label{appendix: algorithm from wine}
\noindent In~\cite{D: anonymous 3}, Theorem~\ref{theorem: structural result for approximate equilibria} was used to design an oblivious PTAS for $n$-player $2$-strategy anonymous games running in time
$${\rm poly}(n) \cdot (1/\epsilon)^{O(1/\epsilon^2)} \cdot U,$$
where $U$ is the number of bits required to represent a payoff value of the game. The algorithm has the following structure

\begin{figure}[h]
\centering
\framebox{
\begin{minipage}{14cm}
\begin{enumerate}

\item choose $k = O(1/\epsilon)$, according to Theorem~\ref{theorem: structural result for approximate equilibria}; \label{item in algorithm: choosing your $k$}

\item guess the number $t$ of players who randomize, the number $t_0$ of players who play $0$, and the number $t_1=n-t-t_0$ of players who play $1$; \label{item in algorithm: how many players play 0 and 1}

\item depending on the number $t$ of players who mix try one of the following:
\begin{enumerate}
\item if $t \le k^3$, guess the number of players $\psi_i$ playing each of the integer multiples $i/k^2$ of $\frac{1}{k^2}$, and, solving a max-flow instance (see details in~\cite{D: anonymous 3}), check if there is a $\epsilon$-Nash equilibrium in which $\psi_i$ players play $i / k^2$, for all $i$, $t_0$ players play $0$, and $t_1$ players play $1$; \label{item in algorithm: exhaustive search over multiples of 1/k squared}

\item if $t > k^3$, guess an integer multiple $i/kn$ of $1/kn$ and, solving a max-flow instance, check if there is an $\epsilon$-Nash equilibrium in which $t$ players play $i/kn$, $t_0$ players play $0$, and $t_1$ players play $1$. \label{item in algorithm: number of players is large}
\end{enumerate}
\end{enumerate}
\end{minipage}
}
\caption{The oblivious PTAS of [Daskalakis, 2008]}
\label{fig: the oblivious PTAS of dask}
\end{figure}

Clearly, there are $O(n^2)$ possible choices for Step~\ref{item in algorithm: how many players play 0 and 1} of the algorithm. Moreover, the search of Step~\ref{item in algorithm: number of players is large} can be completed in time (see~\cite{D: anonymous 3})
$$U\cdot {\rm poly}(n) \cdot (1/\epsilon) \log_2(1/\epsilon),$$
which is polynomial in $1/\epsilon$. On the other hand, Step~\ref{item in algorithm: exhaustive search over multiples of 1/k squared} involves searching over all partitions of $t$ balls into $k^2-2$ bins. The resulting running time for this step (see details in~\cite{D: anonymous 3}) is
$$U\cdot {\rm poly}(n) \cdot (1/\epsilon)^{O(1/\epsilon^2)},$$
which is exponential in $1/\epsilon$.

\subsection{{\large{\sc Moment Search}} : Missing Details} \label{appendix: classifier for moment search}
We describe in detail the third step of {\sc Moment Search}.
\begin{enumerate}
\setcounter{enumi}{2}
\item
For each player $i=1,\ldots, n$, find a subset $$\mathcal{S}_i \subseteq \left\{0, \frac{1}{k^2},\ldots,\frac{k^2-1}{k^2},1\right\}$$ of permitted mixed strategies for that player in an $\epsilon$-Nash equilibrium, ``conditioning'' on the total number of players playing $0$ being $t_0$, the total number of players playing $1$ being $t_1$, and the probabilities of the players who mix resulting in the power-sums $\mu_1,\ldots,\mu_d$ and $\mu'_1,\ldots,\mu'_d$. The way we compute the set $\mathcal{S}_i$ is as follows: 
\begin{enumerate}
\item To determine whether $0 \in \mathcal{S}_i$:
\begin{enumerate}
\item Find {\em any} set of mixed strategies $q_1,\ldots,q_{t_s} \subseteq \{\frac{1}{k^2},\frac{2}{k^2},\ldots,\frac{1}{2}\}$ such that $\sum_{i=1}^{t_s} q_i^{\ell} = \mu_{\ell}$, for all $\ell=1,\ldots,d$. Find {\em any} set of mixed strategies $r_1,\ldots,r_{t_b} \subseteq \{\frac{1}{2}+\frac{1}{k^2},\frac{1}{2}+\frac{2}{k^2},\ldots,1-\frac{1}{k^2}\}$ such that $\sum_{i=1}^{t_b} r_i^{\ell} = \mu'_{\ell}$, for all $\ell=1,\ldots,d$. If such values do not exist {\sc Fail}.\\\text{}\\
{\bf Remark:} An efficient algorithm to solve this optimization problem is described in the proof of Claim~\ref{claim: easy to solve moment equations}.\\
\label{item in Moment Search: determining whether probability 0 is in Si: definition of variable q's and r's}

\item Define the random variable $$Y= (t_0-1)\cdot 0 + \sum_{i=1}^{t_s} S_i + \sum_{i=1}^{t_b} B_i + t_1 \cdot 1,$$
where the variables $S_1,\ldots,S_{t_s}, B_1,\ldots,B_{t_b}$ are mutually independent with expectations $\E[S_i]=q_i$, for all $i=1,\ldots,t_s$, and $\E[L_i]=r_i$, for all $i=1,\ldots,t_b$.

\item Compute the expected payoff $\U^i_0=\E[u^i_0(Y)]$ and $\U^i_1=\E[u^i_1(Y)]$ of player $i$ for playing $0$ and $1$ respectively, if the number of the other players playing $1$ is distributed identically to $Y$. \label{item in Moment Search: determining whether probability 0 is in Si: computing expected utilities}

\item if $\U^i_0 \ge \U^i_1-3\epsilon/4$, then include $0$ to the set $\mathcal{S}_i$, otherwise do not.
\end{enumerate}
\item To determine whether $1 \in \mathcal{S}_i$, follow the same strategy except now $Y$ is defined as follows 
$$Y= t_0 \cdot 0 + \sum_{i=1}^{t_s} S_i + \sum_{i=1}^{t_b} B_i + (t_1-1) \cdot 1,$$
to account for the fact that we are testing for the candidate strategy $1$ for player $i$. Also, the test  that determines whether $1 \in \mathcal{S}_i$ is now $\U^i_1 \ge \U^i_0-3\epsilon/4$. 

\item For all $j \in \{1,\ldots,k^2/2\}$, 
to determine whether $j/k^2 \in \mathcal{S}_i$ do the following slightly updated test: \label{item in Moment Search: determining whether a probability in 0-1/2 is in Si}
\begin{enumerate}
\item Find any set of mixed strategies $q_1,\ldots,q_{t_s-1} \subseteq \{\frac{1}{k^2},\frac{2}{k^2},\ldots,1/2\}$ such that $\sum_{i=1}^{t_s-1} q_i^{\ell} = \mu_{\ell} - (j/k^2)^{\ell}$, for all $\ell=1,\ldots,d$. Find any set of mixed strategies $$r_1, \ldots, r_{t_b} \subseteq \{\frac{1}{2}+\frac{1}{k^2}, \frac{1}{2}+\frac{2}{k^2}, \ldots, 1-\frac{1}{k^2}\}$$ such that $\sum_{i=1}^{t_b} r_i^{\ell} = \mu'_{\ell}$, for all $\ell=1,\ldots,d$. If such values do not exist {\sc Fail}. \label{item in Moment Search: determining whether a probability in 0-1/2 is in Si: definition of variable q's and r's}

\item Define the random variable 
$$Y= t_0 \cdot 0 + \sum_{i=1}^{t_s-1} S_i + \sum_{i=1}^{t_b} B_i + t_1 \cdot 1,$$
where the variables $S_1,\ldots,S_{t_s-1}, B_1,\ldots,B_{t_b}$ are mutually independent with $\E[S_i]=q_i$, for all $i=1,\ldots,t_s-1$, and $\E[L_i]=r_i$, for all $i=1,\ldots,t_b$.\label{item in Moment Search: determining whether a probability in 0-1/2 is in Si: definition of variable Y}

\item Compute the expected payoff $\U^i_0=\E[u^i_0(Y)]$ and $\U^i_1=\E[u^i_1(Y)]$ of player $i$ for playing $0$ and $1$ respectively, if the number of the other players playing $1$ is distributed identically to $Y$.~\label{item in Moment Search: determining whether a probability in 0-1/2 is in Si: computing expected utilities}

\item if $\U^i_0 \in [\U^i_1-3\epsilon/4,\U^i_1+3\epsilon/4]$, then include $j/k^2$ to the set $\mathcal{S}_i$, otherwise do not. \label{item in Moment Search: determining whether a probability in 0-1/2 is in Si: do i include j/(k squared) in set Si}
\end{enumerate}

\item For all $j \in \{(k^2+2)/2,\ldots,k^2-1\}$, 
to determine whether $j/k^2 \in \mathcal{S}_i$ do the appropriate modifications to the method described in Step~\ref{item in Moment Search: determining whether a probability in 0-1/2 is in Si}.
\end{enumerate}
\end{enumerate}

\subsection{The Analysis of {\large \sc Moment Search}} \label{appendix: analysis of moment search}

\paragraph{Correctness}
The correctness of {\sc Moment Search} follows from the following two claims.

\begin{claim} \label{claim: moment search never fails}
If there exists an $\epsilon$/2-Nash equilibrium in which $t \le k^3$ players mix, and their mixed strategies are integer multiples of $1/k^2$, then {\sc Moment Search} will not fail, i.e. it will output a set of mixed strategies $(v_1,\ldots,v_n)$.
\end{claim}

\begin{claim} \label{claim: moment search always outputs equilibrium}
If {\sc Moment Search} outputs a set of mixed strategies $(v_1,\ldots,v_n)$, then these strategies constitute an $\epsilon$-Nash equilibrium.
\end{claim}

\begin{prevproof}{Claim}{claim: moment search never fails}
Let $(p_1,\ldots,p_n)$ be an $\epsilon/2$-Nash equilibrium in which
$t_0$ players play $0$, $t_1$ players play $1$, and $t \le k^3$ players mix,
and their mixed strategies are integer multiples of $1/k^2$. It
suffices to show that there exist guesses for $t_0$, $t_1$, $t_s$, $t_b$,
$\mu_1,\ldots,\mu_d$, $\mu'_1,\ldots,\mu'_d$, such that $p_1 \in
\mathcal{S}_1$, $p_2 \in \mathcal{S}_2$,$\ldots$, $p_n \in
\mathcal{S}_n$. Indeed, let
$$\I_0 := \{i| p_i=0 \},~\I_s := \{i| p_i \in (0,1/2] \},$$
$$\I_b := \{i| p_i \in (1/2,1) \},~\I_1 := \{i| p_i=1 \},$$
and let us choose the following values for our guesses
$$t_0:=|\I_0|, t_s = |\I_s|,~t_b = |\I_b|, t_1:=|\I_1|$$
and, for all $\ell \in [d]$,
$$\mu_{\ell}=\sum_{i \in \I_s} p_i^{\ell},~~\mu'_{\ell}=\sum_{i \in \I_b} p_i^{\ell}.$$
We will show that for the guesses that we defined above $p_i \in
\mathcal{S}_i$, for all $i$. We distinguish the following cases: $i
\in \I_0$, $i \in \I_s$, $i \in \I_b$, $i \in \I_1$. The proof for
all the cases proceeds in the same fashion. We will only argue about
the case $i \in I_s$; in particular, we will show that in
Step~\ref{item in Moment Search: determining whether a probability
in 0-1/2 is in Si} of {\sc Moment Search} the test succeeds for
$j/k^2=p_i$.

At the equilibrium point $(p_1,\ldots,p_n)$, the number of the other
players who choose strategy $1$, from the perspective of player $i$,
is distributed identically to the random variable:
$$Z:= \sum_{j \in \I_s \setminus \{i\}} X_j + \sum_{j \in \I_b} X_j + t_1 \cdot 1,$$
where $\E[X_j]=p_j$ for all $j$. Since $(p_1,\ldots,p_n)$ is an
$\epsilon/2$-Nash equilibrium it must be the case that
\begin{align}
|\E[u^i_0(Z)] - \E[u^i_1(Z)]| \le \epsilon/2. \label{eq: proof of
Moment Search: what I get from Nash}
\end{align}
We will argue that, if in the above equation, we replace $Z$ by $Y$,
where $Y$ is the random variable defined in Step~\ref{item in Moment
Search: determining whether a probability in 0-1/2 is in Si:
definition of variable Y} of {\sc Moment Search}, the inequality
still holds with slightly updated upper bound:
\begin{align}
|\E[u^i_0(Y)] - \E[u^i_1(Y)]| \le 3\epsilon/4. \label{eq: proof of
Moment Search: target equation}
\end{align}
If~\eqref{eq: proof of Moment Search: target equation} is
established, the proof is completed since Step~\ref{item in Moment
Search: determining whether a probability in 0-1/2 is in Si: do i
include j/(k squared) in set Si} will include $j/k^2$ into the set
$\mathcal{S}_i$.

Let $S_1,\ldots,S_{t_s-1}, B_1,\ldots,B_{t_b}$ be the random
variables with expectations $q_1,\ldots,q_{t_s-1},
r_1,\ldots,r_{t_b}$ defined in Step~\ref{item in Moment Search:
determining whether a probability in 0-1/2 is in Si: definition of
variable Y} of {\sc Moment Search}. Observe that, for all
$\ell=1,\ldots,d$,
$$\sum_{j=1}^{t_s-1} q_j^{\ell} = \mu_{\ell} - (j/k^2)^{\ell} = \sum_{j \in \I_s \setminus \{i\}} p_j^{\ell},$$
since $p_i = j/k^2$. Hence, by Theorem~\ref{theorem:binomial appx
theorem},
\begin{align}
\left\| \sum_{j=1}^{t_s-1}S_j - \sum_{j \in \I_s \setminus
\{i\}}X_j\right\| \le 20(d+1)^{1/4} 2^{-(d+1)/2} \le \epsilon/16.
\label{eq: proof of Moment Search: small guys}
\end{align}
Via similar arguments and Corollary~\ref{theorem:binomial appx
theorem for large guys}, we get
\begin{align}
\left\| \sum_{j=1}^{t_b}B_j - \sum_{j \in \I_b}X_j\right\| \le
\epsilon/16. \label{eq: proof of Moment Search: large guys}
\end{align}
\eqref{eq: proof of Moment Search: small guys} and \eqref{eq: proof
of Moment Search: large guys} imply via the coupling lemma
\begin{align}
\| Y ; Z \| \le \frac{\epsilon}{8}. \label{eq: proof of Moment
Search: total variation}
\end{align}
It is not hard to see that
$$|\E[u^i_0(Y)]-\E[u^i_0(Z)]| \le \| Y ; Z \| \le \frac{\epsilon}{8},$$
where we used \eqref{eq: proof of Moment Search: total variation}.
Similarly,
$$|\E[u^i_1(Y)]-\E[u^i_1(Z)]| \le \| Y ; Z \| \le \frac{\epsilon}{8}.$$
Combining the above with \eqref{eq: proof of Moment Search: what I
get from Nash} we get \eqref{eq: proof of Moment Search: target
equation}. This concludes the proof.
\end{prevproof}

\begin{prevproof}{Claim}{claim: moment search always outputs equilibrium}
Let
$$\I_0 := \{i| v_i=0 \},~\I_s := \{i| v_i \in (0,1/2] \},$$
$$\I_b := \{i| v_i \in (1/2,1) \},~\I_1 := \{i| v_i=1 \},$$
$$t_s = |\I_s|,~t_b = |\I_b|,$$
Observe that the moment values that were guessed in Step~\ref{item
in Moment Search: guessing moments} of {\sc Moment Search} satisfy
$$\mu_{\ell}=\sum_{i \in \I_s} v_i^{\ell},~~\mu'_{\ell}=\sum_{i \in \I_b} v_i^{\ell},~~\text{for all $\ell=1,\ldots,d$.}$$

We will argue that $(v_1,\ldots,v_n)$ is an $\epsilon$-Nash
equilibrium. To do this we need to argue that, for each player $i$,
$v_i$ is an $\epsilon$-well supported strategy against the
strategies of her opponents. We distinguish the following cases: $i
\in \I_0$, $i \in \I_s$, $i \in \I_b$, $i \in \I_1$. The proof for
all the cases proceeds in a similar fashion. We will only present
the argument for the case $i \in I_s$.

Let $v_i = j/k^2$ for some $j \in \{1,\ldots,\frac{k^2}{2}\}$. From
the perspective of player $i$, the number of other players who play
$1$ in the mixed strategy profile $(v_1,\ldots,v_n)$ is distributed
identically to the random variable
$$Z:= \sum_{j \in [n] \setminus \{i\}} X_j,$$
where $\E[X_j]=v_j$ for all $j$. To argue that $v_i$ is an
$\epsilon$-well supported strategy against the strategies of $i$'s
opponents, we need to show that
\begin{align}
|\E[u^i_0(Z)] - \E[u^i_1(Z)]| \le \epsilon. \label{eq: proof of
Moment Search 2: what I need to show}
\end{align}

Let us now go back to the iteration of Step~\ref{item in Moment
Search: determining whether a probability in 0-1/2 is in Si} in
which the probability value $j/k^2$ was inserted into the set
$\mathcal{S}_i$. Let $q_1,\ldots,q_{t_s-1}$, $r_1,\ldots,r_{t_b}$ be
the values that were selected at Step \ref{item in Moment Search:
determining whether a probability in 0-1/2 is in Si: definition of
variable q's and r's} of that iteration, and let
$$Y= \sum_{j=1}^{t_s-1} S_j + \sum_{j=1}^{t_b} B_j + t_1 \cdot 1,$$
be the random variable defined in Step~\ref{item in Moment Search:
determining whether a probability in 0-1/2 is in Si: definition of
variable Y}, where the variables $S_1,\ldots,S_{t_s-1},
B_1,\ldots,B_{t_b}$ are mutually independent with expectations
$\E[S_i]=q_i$, for all $i=1,\ldots,t_s-1$, and $\E[B_i]=r_i$, for
all $i=1,\ldots,t_b$. Observe that the $q_j$'s and $r_j$'s where
chosen by Step~\ref{item in Moment Search: determining whether a
probability in 0-1/2 is in Si: definition of variable q's and r's}
so that the following are satisfied

\begin{align}
\sum_{j=1}^{t_s-1} q_j^{\ell} = \mu_{\ell} - (j/k^2)^{\ell} &= \mu_{\ell} - v_i^{\ell} =  \sum_{j \in \I_s \setminus \{i\}} v_j^{\ell},~\text{for all $\ell \in [d]$,} \label{eq: proof of Moment Search 2: agreement in moments by small guys}\\
\text{and~~~}\sum_{j=1}^{t_b} r_j^{\ell} &= \mu'_{\ell} = \sum_{j
\in \I_b} v_j^{\ell},~~~\text{for all $\ell=1,\ldots,d$.}\label{eq:
proof of Moment Search 2: agreement in moments by large guys}
\end{align}
Equation~\eqref{eq: proof of Moment Search 2: agreement in moments
by small guys} implies via Theorem~\ref{theorem:binomial appx
theorem} that
\begin{align}
\left\| \sum_{j=1}^{t_s-1}S_j - \sum_{j \in \I_s \setminus
\{i\}}X_j\right\| \le 20(d+1)^{1/4} 2^{-(d+1)/2} \le \epsilon/16.
\label{eq: proof of Moment Search2: small guys}
\end{align}
Equation~\eqref{eq: proof of Moment Search 2: agreement in moments
by large guys} and Corollary~\ref{theorem:binomial appx theorem for
large guys} imply
\begin{align}
\left\| \sum_{j=1}^{t_b}B_j - \sum_{j \in \I_b}X_j\right\| \le
\epsilon/16. \label{eq: proof of Moment Search2: large guys}
\end{align}
\eqref{eq: proof of Moment Search2: small guys} and \eqref{eq: proof
of Moment Search2: large guys} imply via the coupling lemma
\begin{align}
\| Y ; Z \| \le \frac{\epsilon}{8}. \label{eq: proof of Moment
Search2: total variation}
\end{align}
It is not hard to see that
\begin{align}
|\E[u^i_0(Y)]-\E[u^i_0(Z)]| \le \| Y ; Z \| \le \frac{\epsilon}{8},
\label{eq: proof of Moment Search2: difference in utility for 0}
\end{align}
where we used \eqref{eq: proof of Moment Search2: total variation}.
Similarly,
\begin{align}
|\E[u^i_1(Y)]-\E[u^i_1(Z)]| \le \| Y ; Z \| \le \frac{\epsilon}{8}.
\label{eq: proof of Moment Search2: difference in utility for 1}
\end{align}
Moreover, notice that the random variable $Y$ satisfies the
following condition
\begin{align}
|\E[u^i_0(Y)] - \E[u^i_1(Y)]| \le 3\epsilon/4, \label{eq: proof of
Moment Search2: what I get from the success of moment search}
\end{align}
since, in order for $v_i$ to be included into $\mathcal{S}_i$, the
test in Step~ \ref{item in Moment Search: determining whether a
probability in 0-1/2 is in Si: do i include j/(k squared) in set Si}
of {\sc Moment Search} must have succeeded. Combining \eqref{eq:
proof of Moment Search2: difference in utility for 0}, \eqref{eq:
proof of Moment Search2: difference in utility for 1} and \eqref{eq:
proof of Moment Search2: what I get from the success of moment
search} we get~\eqref{eq: proof of Moment Search 2: what I need to
show}. This concludes the proof.
\end{prevproof}

\paragraph{Computational Complexity}

We will argue that there exists an implementation of {\sc Moment Search} which on input $k=O(1/\epsilon)$ runs in time
$$U\cdot {\rm poly}(n) \cdot (1/\epsilon)^{O( \log^2 (1/\epsilon))},$$
where $U$ is the number of bits required to represent a payoff value of the game.

Observe first that the number of possible guesses for Step~\ref{item in Moment Search: guessing how many small how many large} of {\sc Moment Search} is at most $n^2 O((1/ \epsilon)^6)$. Observe further that the number of possible guesses for $\mu_{\ell}$ in Step~\ref{item in Moment Search: guessing moments} is at most $t \left({k^2 \over 2}\right)^{\ell}$ (where $t \le k^3$ is the number of players who mix), so jointly the number of possible guesses for $\mu_1,\ldots,\mu_{d}$ is at most
$$\prod_{\ell=1}^d{t \left({k^2 \over 2}\right)^{\ell}} = t^d {\left({k^2 \over 2}\right)^{d(d+1)/2}} = \left(1 \over \epsilon \right)^{O\left(\log^2 {1 \over \epsilon}\right)}.$$
The same asymptotic upper bound applies to the total number of guesses for $\mu'_1,\ldots,\mu'_{d}$. Given the above the total number of guesses that {\sc Moment Search} has to do is
$$n^2 \left(1 \over \epsilon \right)^{O\left(\log^2 {1 \over \epsilon}\right)}.$$

We next argue that the running time required to complete Steps~\ref{item in Moment Search: determining the sets Si},~\ref{item in Moment Search: check if there is an assignment from the Sis implementing the moments}, and~\ref{item in Moment Search: output a mixed strategy profile} is at most
$$O(n^3) \cdot U \cdot \left({1 \over \epsilon}\right)^{O(\log^2(1/ \epsilon))}.$$
For this we establish the following; we give the proof in the end of this section.

\begin{claim} \label{claim: easy to solve moment equations}
Given a set of values $\mu_1,\ldots,\mu_d, \mu'_1,\ldots,\mu'_d$, where, for all $\ell=1,\ldots,d$,
$$\mu_{\ell},\mu'_{\ell} \in \left\{0,\left({1 \over k^2}\right)^{\ell},2\left({1 \over k^2}\right)^{\ell},\ldots,B\right\},$$
for some $B \in \mathbb{N}$, discrete sets $\mathcal{T}_1,\ldots,\mathcal{T}_m \subseteq \left\{0, {1 \over k^2}, {2 \over k^2},\ldots, 1\right\}$, and four integers $m_0, m_1 \le m$, $m_s, m_b \le B$, it is possible to solve the system of equations:
\begin{align*}
(\Sigma):~~~\sum_{p_i \in (0,1/2]}p_i^{\ell}&=\mu_{\ell}, \text{ for all $\ell=1,\ldots,d$},\\
\sum_{p_i \in (1/2,1)}p_i^{\ell}&=\mu'_{\ell}, \text{ for all $\ell=1,\ldots,d$},\\
|\{i | p_i = 0\}|&= m_0\\
|\{i | p_i = 1\}|&=m_1\\
|\{i | p_i \in (0, 1/2]\}|&=m_s\\
|\{i | p_i \in (1/2, 1)\}|&=m_b\\
\end{align*}
 with respect to the variables $p_1 \in \mathcal{T}_1,\ldots, p_m \in \mathcal{T}_m$, or to determine that no solution exists, in time
$$O(m^3) B^{O(d)} k^{O(d^2)}.$$
\end{claim}

Applying Claim~\ref{claim: easy to solve moment equations} with $m \le t$, $B \le t$ (where $t \le k^3$ is the number of players who mix), $m_0=0$, $m_1=0$, shows that Steps~\ref{item in Moment Search: determining whether probability 0 is in Si: definition of variable q's and r's},~\ref{item in Moment Search: determining whether a probability in 0-1/2 is in Si: definition of variable q's and r's} can be completed in time
$$O(t^3) t^{O(d)} k^{O(d^2)} = \left(1\over \epsilon\right)^{O(\log^2(1/\epsilon))}.$$
Another application of Claim~\ref{claim: easy to solve moment equations} with $m = n$, $B \le t$, $m_0 \le n$, $m_1 \le n$ shows that Step~\ref{item in Moment Search: check if there is an assignment from the Sis implementing the moments} of {\sc Moment Search} can be completed in time
$$O(n^3) t^{O(d)} k^{O(d^2)} = O(n^3) \cdot \left(1\over \epsilon\right)^{O(\log^2(1/\epsilon))}.$$
Finally, we argue that the computation of the expected utilities $\U^i_0$ and $\U^i_1$ required in Steps~\ref{item in Moment Search: determining whether probability 0 is in Si: computing expected utilities}, \ref{item in Moment Search: determining whether a probability in 0-1/2 is in Si: computing expected utilities} of {\sc Moment Search} can be done efficiently using dynamic programming with $O(n^2)$ operations on numbers with at most $b(n,k):=\lceil 1+n \log_2{(k^2)} + U) \rceil$ bits, where $U$ is the number of bits required to specify a payoff value of the game.\footnote{To compute a bound on the number of bits required for the expected utility computations, note that every non-zero probability value that is computed along the execution of the algorithm must be an integer multiple of $(\frac{1}{k^2})^{n-1}$, since the mixed strategies of all players are from the set $\{0,1/k^2,2/k^2,\ldots,1\}$. Further note that the expected utility is a weighted sum of $(n-1)$ payoff values, with $U$ bits required to represent each value, and all weights being probabilities.}

Therefore, the overall time required for the execution of {\sc Moment Search} is
$$O(n^3) \cdot U \cdot \left(1\over \epsilon\right)^{O(\log^2(1/\epsilon))}.$$

\begin{prevproof}{Claim}{claim: easy to solve moment equations}
We use dynamic programming. Let us consider the following tensor of
dimension $2d+5$:
$$A(i, z_0, z_1, z_s, z_b; \nu_1,\ldots,\nu_d ; \nu'_1,\ldots,\nu'_d),$$
where $i \in [m]$, $z_0, z_1 \in \{0,\ldots,m\}$, $z_s, z_b \in \{0,\ldots,B\}$ and $$\nu_{\ell},
\nu'_{\ell} \in \left\{0,\left({1 \over k^2}\right)^{\ell},2\left({1
\over k^2}\right)^{\ell},\ldots,B\right\},$$ for $\ell=1,\ldots,d.$
The total number of cells in $A$ is
\begin{align*}&m \cdot (m+1)^2 \cdot (B+1)^2 \cdot \left( \prod_{\ell =1}^d (B k^{2 \ell}+1) \right)^2\\&~~~~~~~~~~~~~~~~\le O(m^3) B^{O(d)} k^{2d(d+1)}.
\end{align*}
Every cell of $A$ is assigned value $0$ or $1$, as follows:
\begin{align*}
&A(i, z_0, z_1, z_s, z_b ; \nu_1,\ldots,\nu_d, \nu'_1,\ldots,\nu'_d)=1\\&~~\Leftrightarrow \left(\begin{minipage}{6.5cm} \centering There exist $p_1 \in \mathcal{T}_1$, $\ldots$, $p_i \in \mathcal{T}_i$ such that $|\{j \le i | p_j =0\}|=z_0$, $|\{j \le i | p_j =1\}|=z_1,$ $|\{j \le i | p_j \in (0,1/2]\}|=z_s$, $|\{j \le i | p_j \in (1/2,1)\}|=z_b,$
$\sum_{j \le i: p_j \in (0,1/2]}p_j^{\ell}=\nu_{\ell}$, for all
$\ell=1,\ldots,d$, $\sum_{j \le i: p_j \in
(1/2,1)}p_j^{\ell}=\nu'_{\ell}$, for all
$\ell=1,\ldots,d$.\end{minipage} \right).
\end{align*} It is easy to complete
$A$ working in layers of increasing $i$. We initialize all entries
to value $0$. Then, the first layer
$A(1,\cdot,\cdot~;~\cdot,\ldots,\cdot)$ can be completed easily as
follows:
\begin{align*}
&A(1, 1, 0, 0, 0 ; 0, 0, \ldots,0 ; 0, 0, \ldots, 0)=1 \Leftrightarrow 0 \in \mathcal{T}_1\\
&A(1, 0, 1, 0, 0 ; 0, 0, \ldots,0 ; 0, 0 \ldots, 0)=1 \Leftrightarrow 1 \in \mathcal{T}_1\\
&A(1, 0, 0, 1, 0 ; p, p^2, \ldots,p^d ; 0,\ldots,0)=1 \Leftrightarrow p \in \mathcal{T}_1 \cap (0,1/2]\\
&A(1, 0, 0, 0, 1 ; 0, \ldots,0; p, p^2, \ldots,p^d)=1 \Leftrightarrow p \in \mathcal{T}_1 \cap (1/2,1)
\end{align*}
Inductively, to complete layer $i+1$, we consider all the non-zero
entries of layer $i$ and for every such non-zero entry and for every
$v_{i+1} \in \mathcal{T}_{i+1}$, we find which entry of layer $i+1$
we would transition to if we chose $p_{i+1}=v_{i+1}$. We set that
entry equal to $1$ and we also save a pointer to this entry from the
corresponding entry of layer $i$, labeling that pointer with the
value $v_{i+1}$. The time we need to complete layer $i+1$ is bounded
by
$$|\mathcal{T}_{i+1}|. (m+1)^2 B^{O(d)} k^{2d(d+1)} \le O(m^2) B^{O(d)} k^{O(d^2)}.$$
Therefore, the overall time needed to complete $A$ is $$O(m^3)
B^{O(d)} k^{O(d^2)}.$$

After completing tensor $A$, it is easy to check if there exists a
solution to $(\Sigma)$. A solution exists if and only if
$$A(m,m_0,m_1,m_s,m_b;\mu_1,\ldots,\mu_d ; \mu'_1,\ldots,\mu'_d)=1,$$ and can
be found by tracing back the pointers from this cell of $A$. The
overall running time is dominated by the time needed to fill in $A$.
\end{prevproof}}
\end{document}